%% file: BDG-CCMB-LMCS-revision.tex
\documentclass{lmcs}
\pdfoutput=1

\usepackage[utf8]{inputenc}
\usepackage{lastpage}
\lmcsdoi{16}{3}{3}
\lmcsheading{}{\pageref{LastPage}}{}{}%
{May~03,~2019}{Jul.~10,~2020}{}


\RequirePackage{doi}

\RequirePackage{hyperref}

\input{BDG-LMCS--preamble.tex} 

\begin{document}

\title[Combinatorial Conversion and Moment Bisimulation]{Combinatorial Conversion and Moment Bisimulation\newline for Stochastic Rewriting Systems}

\author[N.~Behr]{Nicolas Behr\rsuper{a}}	
\address{\lsuper{a}IRIF, Universit\'{e} Paris-Diderot, F-75205 Paris Cedex 13, France}	
\email[Corresponding author]{nicolas.behr@irif.fr}  
\thanks{The work of N.B.\@ is supported by a \emph{Marie Sk\l{}odowska-Curie Individual Fellowship} (Grant Agreement No.~753750 -- RaSiR)}	

\author[V.~Danos]{Vincent Danos\rsuper{b}}
\address{\lsuper{b}D\'{e}partement d'Informatique de l'ENS, ENS, CNRS, PSL University Paris, France}
\email{\{vincent.danos,ilias.garnier\}@di.ens.fr}

\author[I.~Garnier]{Ilias Garnier\rsuper{b}}
\thanks{The work of I.G.\@ is supported by the ANR grant REPAS}	





\begin{abstract}
  \noindent We develop a novel method to analyze the dynamics of stochastic rewriting systems evolving over finitary adhesive, extensive categories. Our formalism is based on the so-called rule algebra framework~\cite{bdg2016,Behr2018} and exhibits an intimate relationship between the combinatorics of the rewriting rules (as encoded in the rule algebra) and the dynamics which these rules generate on observables (as encoded in the stochastic mechanics formalism). We introduce the concept of combinatorial conversion, whereby under certain technical conditions the evolution equation for (the exponential generating function of) the statistical moments of observables can be expressed as the action of certain differential operators on formal power series. This permits us to formulate the novel concept of moment-bisimulation, whereby two dynamical systems are compared in terms of their evolution of sets of observables that are in bijection. In particular, we exhibit non-trivial examples of graphical rewriting systems that are moment-bisimilar to certain discrete rewriting systems (such as branching processes or the larger class of stochastic chemical reaction systems). Our results point towards applications of a vast number of existing well-established exact and approximate analysis techniques developed for chemical reaction systems to the far richer class of general stochastic rewriting systems.
\end{abstract}

\maketitle

\section{Introduction}

Stochastic graph rewriting systems (SGRSs) have proved their adequacy at modeling various phenomena~\cite{nagl1978tutorial,heckelgraph,ehrig2008graph},
ranging from protein-protein interactions in
biology~\cite{danos2004formal,danos2012graphs} to social network dynamics~\cite{gross2009adaptive,pinaud2012porgy}. In this regard, SGRSs provide a promising setting for the development of theoretical and algorithmic tools for analyzing the properties of a wide range of models. These properties correspond to \emph{observables}, i.e.\ functions that describe the time-dependent occurrences of patterns of interest in the system being studied. The problem of computing these observables, either exactly or approximately, has garnered a lot of attention in recent years~\cite{danos2004computational,danos2008rule,Feret2012137,Petrov_2012,danos2015moment}.

In this paper, we approach this problem by formulating SGRSs in the \emph{rule-algebraic framework}~\cite{bdg2016} and address the problem for a far greater class of \emph{stochastic rewriting systems} (SRSs) using a recent extension to rewriting in adhesive, extensive categories~\cite{Behr2018}. Notably, the rule algebras are the associative unital algebras of compositions of rewriting rules, from which the action of rewriting rules on objects of the underlying category may be reconstructed in terms of the canonical representation of the rule algebras. From this point, well-established techniques from the realm of stochastic mechanics permit to develop a well-founded implementation of SRSs.  We achieve the following new results. Firstly, we prove a theorem allowing one to extract from a given SRS and a chosen family of observables
(satisfying a certain closure property w.r.t. the dynamics) a partial differential equation (PDE) describing the time-evolution of the \emph{exponential moment generating function} (EMGF) of that family. We refer to this procedure as \emph{combinatorial conversion}, since the resulting evolution equations are determined solely from algebraic combinatorial relations. The mathematical setting for this evolution is the realm of formal power series, and one can leverage powerful tools from abstract algebra and analytic combinatorics to study the higher moments of these families of observables. Secondly, we introduce the notion of \emph{moment bisimulation}. The idea is based on our observation that two given SRSs may exhibit the same EMGF evolution for certain choices of observables even if their underlying transition systems themselves do not coincide. We define a precise mathematical setup to efficiently compare SRSs via their EMGF evolution. The main motivation for this concept is the possibility to study the statistical moments of a set of observables for a complicated SRS (that may itself not even be amenable to efficient simulation) via analyzing a \emph{moment-bisimilar}, potentially less complex system instead. For vertex-counting observables, we obtain an exact characterization of the moment-bisimilarity class of families of SRS which are bisimilar to \emph{discrete graph rewriting systems}. The latter systems are better known as chemical reaction systems (i.e.\@ stochastic Petri nets with mass-action semantics). Our definition and analysis of bisimilarity classes is based upon the evolution equations for exponential moment generating functions of the observables of the stochastic rewriting systems. In particular, we show how to synthesize the corresponding discrete SRS from the information of the exponential moment generating function evolution equations.

\subsection{Relation to previous work}

The present article aims to relate two main fields of research, namely the theory of stochastic rewriting systems and the theory of chemical reaction systems. With regards to the former field, we would like to mention in particular the results that have been obtained in the development and analysis of the so-called Kappa framework~\cite{danos2004formal} that was introduced to study biochemical reaction frameworks. Since Kappa is based on (a particular notion of site-) graph rewriting, its practical development has seen a large number of interesting algorithms and mathematical structures. In terms of approximate and exact analysis, we would like to mention the ideas related to ``granularity''~\cite{Harmer_2010} and ideas of coarse-graining of descriptions of dynamical models (which bear a certain similarity with our notion of moment bisimulation), the analysis of the combinatorial complexity of rewriting systems (in biochemistry)~\cite{Deeds_2012,Abou_Jaoud__2016}, and also the idea of stochastic fragments as introduced in~\cite{feretstochastic} (see also~\cite{Murphy_2010,Petrov_2012}). The latter concept is similar in spirit to our combinatorial conversion results, albeit the precise mathematical setups is notably different (and the formulation in loc cit.\@ is in particular rather specialized to the settings encountered in calculations with labeled transition systems). A large set of results is available in this community regarding the generation of (deterministic) ODEs for expectation values of observables in Kappa systems, including the analyzer framework \textsc{KADE}~\cite{Camporesi_2017} (see also~\cite{danos2004computational} for an overview of earlier work on this topic). Moreover, some interesting results are available~\cite{Danos_2010} that allow one to perform exact or approximate model reduction in terms of the aforementioned ODEs. Our notion of combinatorial conversion and moment bisimilarity could be seen as an \emph{observable-based} generalization of these scenarios, where, instead of focusing on individual moments such as the expectation values of a given set of observables, we give a full equivalence of \emph{all} moments of two given sets of observables in two different models.

In the literature on chemical reaction systems, there is a large set of results available on computability and both analytical and numerical approximate methods in order to tackle studies of the relevant stochastic dynamics. Analytical methods include techniques that rely implicitly on combinatorial semi-linear normal-ordering techniques to treat first order reaction systems~\cite{McQuarrie_1967,jahnke2007solving,bdp2017}, or alternatively rely on the method of characteristics to solve linear differential equations~\cite{reis2018general}, or for certain classes of \emph{chemical reaction systems (CRSs)} methods relying on the notion of (generalized) orthogonal polynomials~\cite{McQuarrie_1967,bdp2017} or special functions~\cite{Hornos:2005aa,Shahrezaei:2008aa,Vandecan:2013aa}. Amongst the perturbative approaches, one finds various notions of finite state projections~\cite{kryven2015solution,cao2016state,dinh2016understanding}, tau-leaping~\cite{cao2006efficient,rathinam2003stiffness,cao2007adaptive}, moment closure techniques~\cite{gillespie2009moment,schnoerr2015comparison}, stochastic path integral techniques~\cite{peliti1985path,ohkubo2012one,Weber_2017}, quantized tensor train formats~\cite{zhang2015enabling,liao2015tensor,gelss2016solving}, and various forms of Magnus expansions~\cite{blanes2009magnus,iserles2018applications}. Many of these techniques might in fact pose interesting approaches also to study the dynamics of more general stochastic rewriting systems, even those that do not possess the special property of discrete bisimilarity to chemical reaction systems. Finally, notions of bisimilarity such as backwards bisimulation based on the concept of lumpability have been studied in~\cite{Feret2012137,cardelli_et_al:LIPIcs:2015:5368}. These methods differ from our approach in that our moment bisimulations  capture dynamical information on a subset of the observables, as opposed to a reduction of information based on coarse-graining the state spaces of the models in the aforementioned methods.

\subsection{Overview of results and plan of the paper}
In Section~\ref{sec:CRS}, we introduce a formulation of chemical reaction systems that will serve as the basis of our general stochastic mechanics formalism and that illustrates some of the main concepts of the theory. The general framework requires the notion of so-called \emph{rule algebras}, i.e.\ of the algebras of sequential compositions of linear rules in rewriting systems over finitary adhesive, extensive categories (with the main concepts and results recalled in Section~\ref{sec:SMF}). As a key result of our formalism, we present in Section~\ref{sec:CC} the notion of \emph{combinatorial conversion}, whereby a given set of observables in a stochastic rewriting system is amenable to an analysis in terms of their exponential moment generating function precisely if it fulfills a certain technical condition that renders the evolution equation for the EMGF into the form of a PDE on certain formal power series. Since the latter no longer explicitly contains information on the precise structure of the observables and of the transitions of the stochastic system, it is natural to define a notion of \emph{moment bisimulation} (see Section~\ref{sec:DB}) as the situation where two equinumerous sets of observables in two different stochastic rewriting system have an isomorphic evolution equation for their EMGFs. Amongst the resulting notion of bisimilarity classes, somewhat surprisingly the bisimilarity class of \emph{discrete graph} rewriting systems (which includes chemical reaction systems) contains also non-trivial classes of more general rewriting systems and sets of observables, which renders the notion of \emph{discrete moment bisimulation} a very interesting candidate for a method of computation for stochastic rewriting systems. We illustrate this phenomenon by providing two application example case studies in Section~\ref{sec:AE} (a variant of the so-called \emph{voter model}) and in Section~\ref{sec:BC} (a \emph{cryptocurrency toy model}).

\section{The blueprint for a general stochastic mechanics framework: chemical reaction systems in the language of Delbr\"uck and Doi}
\label{sec:CRS}

As will become evident in the sequel, chemical reaction systems constitute the simplest possible type of rewriting systems: in the rule algebra framework, they are understood as rewriting systems of \emph{discrete graphs}. To illustrate this claim and to motivate the construction of a general stochastic mechanics framework, let us first consider the case of chemical reaction systems of a single chemical species. We will follow closely our recent exposition~\cite{bdp2017}, which in turn relies on standard chemical reactions systems theory, particularly on the ideas of Delbr\"uck~\cite{delbruck1940statistical} and Doi~\cite{Doi_1976} (see also~\cite{McQuarrie_1967,Weber_2017}).\\

In the formulation in terms of probability generating functions due to Delbr\"uck~\cite{delbruck1940statistical}, a chemical reaction of one species $X$ is a dynamical system with transitions of the form
\begin{equation}\label{eq:CRNtrans}
	i\,X\xrightharpoonup{\kappa_{i,o}}o\, X\,.
\end{equation}
Here, $i$ denotes the number of particles $X$ that enter the transition, $o$ the number that exits the transition, and $\kappa_{i,o}\in \bR_{>0}$ is the so-called base rate of the transition. A \emph{pure state} of the system is characterized as a state with a precise number $n$ of particles $X$, while in general the system will evolve on the space of probability distributions (or \emph{mixed states}) over the pure states. More concretely, starting from a pure state, the probability for a transition as described in~\eqref{eq:CRNtrans} to occur is proportional to its base rate $\kappa_{i,o}$ times the number of possibilities that $i$ particles $X$ can be picked from the $n$ particles in the state (with all particles considered as indistinguishable). This concept is known as \emph{mass action semantics} (see e.g.\ \cite{Gillespie_2007}, or~\cite{Cain_2014} for a short introduction, or~\cite{Voit_2015} for a historical review). Delbr\"uck's key idea consisted in encoding this dynamics via the notion of \emph{probability generating functions}: each pure state of $n$ particles $X$ is encoded as a monomial $x^n$ (with $x$ a formal variable), while a probability distribution at some time $t$ (with $t\geq 0$) over such pure states is mapped to the formal power series
\begin{equation*}
	P(t;x):=\sum_{n\geq 0} p_n(t)x^n\,.
\end{equation*}
Here, $p_n(t)$ denotes the probability at time $t$ that the system is in a state with $n$ particles, and we require the standard conditions
\[
	\forall n\geq0, t\geq 0:\; p_n(t)\in \bR_{\geq0}\;\land\; \forall t\geq 0:\;\sum_{n\geq 0}p_n(t)=1\,.
\]
Then given a set of transitions with base rates as in~\eqref{eq:CRNtrans} and an initial state described by a probability generating function $P_0(x)$, the evolution of the system under mass-action semantics is given in Delbr\"uck's formulation as\footnote{Here, in order to keep the expression for the infinitesimal generator $\cH$ as compact as possible, we set $\kappa_{i,o}=0$ for those transitions not included in the (finite) set of transitions that define the given model.}
\begin{equation}\label{eq:dellEvo}
\begin{aligned}
  \frac{\partial}{\partial t}P(t;x)&= \cH P(t;x)\,,\quad P(0;x)=P_0(x)\\
  \cH&=\sum_{i,o} \kappa_{i,o}\left(\hat{x}^o-\hat{x}^i\right)\left(\tfrac{\partial}{\partial x}\right)^i\,.
\end{aligned}
\end{equation}
The linear operator $\hat{x}$ is defined as the \emph{formal multiplication operator}, $\hat{x}x^n:=x^{n+1}$, while $\frac{\partial}{\partial x}$ denotes the \emph{formal derivative operator}. It is instructive to expand the formal solution of the evolution equation~\eqref{eq:dellEvo}, $P(t;x)=e^{t\cH}P(0;x)$, into a Taylor series around $t=0$:
\[
\begin{aligned}
P(t;x)&=(1+t\cH+O(t^2))P(0;x)\\
&=\left(1-t \sum_{i,o} \kappa_{i,o}\hat{x}^i\left(\tfrac{\partial}{\partial x}\right)^i\right)P(0;x)
+t\sum_{i,o} \kappa_{i,o}\hat{x}^o\left(\tfrac{\partial}{\partial x}\right)^i P(0;x)+O(t^2)\,.
\end{aligned}
\]
Thus, for a given transition from $i$ to $o$ particles $X$ at base rate $\kappa_{i,o}$, the contribution $\kappa_{i,o}\hat{x}^o\partial_x^i$ to $\cH$ implements the transition, while the (diagonal) contribution $-\kappa_{i,o}\hat{x}^i\partial_x^i$ ensures that the normalization of the probability distribution generating function $P(t;x)$ remains intact (see also~\cite{bdp2017} for an extended explanation). Mass-action semantics is indeed implemented via this form of operator $\cH$, since for a given pure state encoded as $x^n$ one finds that
\[
(1+t\cH+O(t^2))x^n= \left((1-\kappa_{i,o}(n)_i)x^n +t\kappa_{i,o}(n)_i\, x^{n-i+o}\right)+O(t^2)\,,
\]
with $(a)_b:=a!/(a-b)!$ (and $(a)_b:=0$ for $b>a$).

\begin{exa}\label{exa:bd}
As a simple, yet sufficiently non-trivial example of a one-species chemical reaction system, let us consider for illustration a \emph{birth-death model} (with $\beta,\tau\in \bR_{\geq0}$):
\begin{equation}\label{eq:bdDef}
    \varnothing\xrightharpoonup{\beta}X\,,\quad X\xrightharpoonup{\tau}\varnothing
\end{equation}
In Delbr\"uck's formalism, we have to solve the evolution equation~\eqref{eq:dellEvo} with infinitesimal generator $\cH$ given as
\begin{equation*}
    \cH=\beta(\hat{x}-1)+\tau(1-\hat{x})\tfrac{\partial}{\partial x}\,.
\end{equation*}
Taking advantage of the fact that $\cH$ is of \emph{semi-linear form} (i.e.\ contains at most a first-order derivative in each term), one may utilize the so-called \emph{semi-linear normal-ordering technique}~\cite{Dattoli:1997iz,blasiak2005boson,blasiak2011combinatorial} in order to obtain a closed-form solution for $P(t;x)$~\cite{bdp2017}:
\begin{equation}\label{eq:solBDP}
\begin{aligned}
    P(t;x)&=Pois(\tfrac{\beta}{\tau}(1-e^{-\tau t});x)\cdot P_0(Bern(e^{-\tau t};x))\\
    Pois(\alpha;x)&:=e^{\alpha(x-1)}\quad (\alpha\in \bR_{\geq0})\,,\quad
    Bern(\alpha;x):=\alpha x+(1-\alpha)\quad (0\leq \alpha\leq 1)
\end{aligned}
\end{equation}
Here, $P_0(x):=P(0;x)$ is the probability generating function (PGF) of the initial state, $Pois(\alpha;x)$ the PGF of a \emph{Poisson distribution of parameter} $\alpha$, and $Bern(\alpha;x)$ that of a \emph{Bernoulli distribution of parameter} $\alpha$. In order to interpret formula~\eqref{eq:solBDP}, let us recall that for any two (discrete) probability distributions $\mu_1$ and $\mu_2$ with respective PGFs $P_1(x)$ and $P_2(x)$, the \emph{product} $P_1(x)\cdot P_2(x)$ yields the PGF of the \emph{convolution} $\mu_1*\mu_2$, while the \emph{composition} $P_1(P_2(x))$ encodes the PGF of the \emph{compound distribution} of $\mu_1$ and $\mu_2$. We thus find that the PGF $P(t;x)$ in~\eqref{eq:solBDP} is a convolution of a Poisson distribution with the compound of the initial state distribution (encoded by $P_0(x)$) and a Bernoulli distribution. Moreover, taking the limit $t\to\infty$, we recover the well-known fact that the system state converges onto a Poisson distribution of parameter $\tfrac{\beta}{\tau}$:
\begin{equation*}
\lim\limits_{t\to\infty}P(t;x)=Pois(\tfrac{\beta}{\tau};x)
\end{equation*}
\end{exa}

In order to extend this formalism to general stochastic rewriting systems, it proves fruitful to consider the following alternative formulation of chemical reaction systems in terms of \emph{boson creation ($a^{\dag}$) and annihilation operators ($a$)} due to M.~Doi~\cite{Doi_1976}. We obtain this formulation immediately from the following isomorphism (of representations of the Heisenberg-Weyl algebra, see e.g.\ \cite{bdp2017}):
\begin{equation}\label{eq:baseChange}
x^n \leftrightarrow \ket{n}\,, \; \hat{x}\leftrightarrow a^{\dag}\,,\; \tfrac{\partial}{\partial x}\leftrightarrow a\,. 
\end{equation}
It is evident from the definition of the operators $a^{\dag}$ and $a$ that they fulfill the so-called \emph{canonical commutation relation}
\begin{equation}\label{eq;HWccr}
[a,a^{\dag}]:=a a^{\dag}-a^{\dag}a = \mathbb{1}\,,
\end{equation}
where $\mathbb{1}$ denotes the identity operator on the vector space spanned by the vectors $\ket{n}$. Moreover, the identities $\hat{x}x^{n}=x^{n+1}$ and $\tfrac{\partial}{\partial x}x^n=nx^{n-1}$ translate into
\begin{equation}\label{eq:HWcanrepClassical}
a^{\dag}\ket{n}=\ket{n+1}\,,\quad a\ket{0}=0\,,\; \forall n>0:\; a\ket{n}=n\ket{n-1}\,.
\end{equation}
Accordingly, the time-dependent state of a chemical reaction system will now be described as a time-dependent state in the basis of so-called \emph{number vectors} $\ket{n}$,
\begin{equation*}
\ket{\Psi(t)}=\sum_{n\geq 0} p_n(t)\ket{n}\,.
\end{equation*}
The operator $\cH$ of~\eqref{eq:dellEvo} is transformed via the change of basis~\eqref{eq:baseChange} into the operator
\begin{equation}\label{eq:HcrnDoi}
H=\sum_{i,o} \kappa_{i,o}\left(a^{\dag\:o}-a^{\dag\:i}\right)a^i\,.
\end{equation}

While thus far this alternative representation does not yield any further insights, a key additional concept in Doi's formalism is the notion of the so-called \emph{dual projection vector} $\bra{}$, which is defined via its action on states as~\cite{Doi_1976} (compare~\cite{bdp2017})
\begin{equation}\label{eq:projBra}
  \forall n\geq 0:\; \braket{}{n}:=1_{\bR}\,.
\end{equation}
One then finds the relationships
\begin{equation*}
\braket{}{\Psi(t)}=1\,,\quad \bra{}H =  0\,,
\end{equation*}
which express the normalization of $\ket{\Psi(t)}$ as a probability distribution, and a certain technical property of $H$ (that is necessary in order for $H$ to qualify as an infinitesimal generator of a continuous-time Markov chain or CTMC). Moreover, one may naturally define the concept of \emph{observables} in this formalism: an observable $O$ is a \emph{diagonal operator} (in the basis $\{\ket{n}\}_{n\geq 0}$),
\begin{equation*}
 \forall n\geq 0:\; O\ket{n}=\omega_O(n)\ket{n}\,.
\end{equation*}
In the basis of linear operators $a^{\dag}$ and $a$, due to the canonical commutation relation~\eqref{eq;HWccr}, any linear operator may be expressed as a linear combination of the \emph{normal-ordered} expressions of the form $a^{\dag\:r}a^s$ (with $r,s\in \bZ_{\geq 0}$). Diagonal linear operators must thus be linear combinations of expressions of the form $a^{\dag\: k}a^k$. Indeed, a classical result from the combinatorics literature (see e.g.~~\cite{blasiak2005combinatorics}) entails that
\begin{equation}\label{eq:discrStirling}
    a^{\dag\: k}a^k=\sum_{n=0}^k s_1(k;n)(a^{\dag}a)^n\,,
\end{equation}
where $s_1(n;k)$ are the (signed) \emph{Stirling numbers of the first kind}~\cite{OEISstirlingFirst}. This identifies the operator $\hat{n}:= a^{\dag}a$, the so-called \emph{number operator} (since $\hat{n}\ket{n}=n\ket{n}$), as the only independent observable for chemical reaction systems of one particle type, as according to~\eqref{eq:discrStirling} every possible observable of a discrete system may be expressed as a polynomial in $\hat{n}$. One might then ask about the \emph{statistical moments} of this observable. Contact with classical probability theory is established via defining the \emph{expectation value} $\bE_{\ket{\Psi(t)}}(O)$ of an observable $O$ in the time-dependent state $\ket{\Psi(t)}$ as
\begin{equation*}
  \bE_{\ket{\Psi(t)}}(O):=\bra{}O\ket{\Psi(t)}\,.
\end{equation*}
Defining the exponential moment generating function (EMGF) $M(t;\lambda)$ of the statistical moments of the number operator $\hat{n}$ as
\begin{equation*}
  M(t;\lambda):=\bra{}e^{\lambda\hat{n}}\ket{\Psi(t)}\,,
\end{equation*}
one may derive the following \emph{EMGF evolution equation} (cf.~\cite{bdp2017}, Thm.~3):
\begin{equation*}
\begin{aligned}
  \tfrac{\partial}{\partial t}M(t;\lambda)&= \bD(\lambda,\partial_{\lambda})M(t;\lambda)\\
  \bD(\lambda,\partial_{\lambda})&= \sum_{i,o}\kappa_{i,o} \left(e^{\lambda(o-i)}-1\right)\sum_{k=0}^{i}s_1(i,k)\left(\tfrac{\partial}{\partial\lambda}\right)^{k}\,.
\end{aligned}
\end{equation*}
It will prove rather important to our main framework to understand the derivation of this result from the structures introduced thus far. Following loc cit.\@ closely, we compute:
\begin{equation}\label{eq:derMEGFch}
\begin{aligned}
		\tfrac{\partial}{\partial t}M(t;\lambda)&=\bra{}e^{\lambda O}H\ket{\Psi(t)}\\
		&=\bra{}\left(e^{\lambda O}He^{-\lambda O}\right)e^{\lambda O}\ket{\Psi(t)}\\
		&\overset{(*)}{=}
			\bra{}\left(e^{ad_{\lambda O}}H\right)e^{\lambda O}\ket{\Psi(t)}\\
		&=\bra{}\left[
			H+\sum_{q\geq1}\frac{\lambda^q}{q!}ad_{O}^{\circ\:q}(H)
		\right]e^{\lambda O}\ket{\Psi(t)}\\
		&\overset{(**)}{=} \sum_{q\geq1}\frac{\lambda^q}{q!}\bra{}\left[ad_{O}^{\circ\:q}(H)
		\right]e^{\lambda O}\ket{\Psi(t)}\,.
\end{aligned}
\end{equation}
Here, in the step marked $(*)$ we have made use of the variant $e^{\lambda A}Be^{-\lambda A}=e^{\lambda ad_A}B$ of the BCH formula (see e.g. \cite{Hall_2015}, Prop.~3.35), where for two composable linear operators $A$ and $B$ the adjoint action is defined as $ad_A B:=AB-BA$, and with $ad_A^{\circ\:q} B$ denoting the $q$-fold application of $ad_A$ to $B$ (and $ad_A^{\circ 0}B:=B$), while the step $(**)$ follows from $\bra{}H = 0$.

\begin{exa}[Ex.~\ref{exa:bd} continued]\label{exa:bd2}
    Following Doi's method, performing the change of basis~\eqref{exa:bd}, the birth-death system of~\eqref{eq:bdDef} is found to have the infinitesimal generator
    \begin{equation*}
        H=\beta(a^{\dag}-1)+\tau(1-a^{\dag})a\,.
    \end{equation*}
    In order to explicitly verify the property $\bra{}H=0$, it is sufficient to verify this property on \emph{basis vectors} $\ket{n}$:
    \begin{equation*}
        \bra{}H\ket{n}=\beta(\braket{}{n+1}-\braket{}{n})+\tau(n\braket{}{n-1}-n\braket{}{n})=0\,.
    \end{equation*}
    As an illustration of the key technique of Doi's method in view of this paper, the calculation of evolution equations for exponential moment-generating functions (EMGFs), let us explicitly carry out the steps of the derivation~\eqref{eq:derMEGFch} for the example at hand. Choosing the observable $\hat{n}:=a^{\dag}a$, let us first compute the following \emph{commutator}:
    \begin{equation*}
        ad_{\hat{n}}(H):=[\hat{n},H]=\hat{n}H-H\hat{n}
        =\beta a^{\dag}-\tau a\,.
    \end{equation*}
    Here, we made use of the fact that for any two \emph{diagonal} operators $A$ and $B$, $[A,B]=0$, while for the last equality we made use of the \emph{canonical commutation relation} $[a,a^{\dag}]=\mathbb{1}$. Inserting this result into~\eqref{eq:derMEGFch}, and with the auxiliary formula
\begin{equation*}
    \sum_{q\geq1}\frac{\lambda^q}{q!}ad_{\hat{n}}^{\circ\:q}(H)
    =\beta(e^{\lambda}-1) a^{\dag}+\tau(e^{-\lambda}-1) a\,,
\end{equation*}
the first and last equations of~\eqref{eq:derMEGFch} evaluate to
\begin{equation}\label{eq:MEGFbd}
\begin{aligned}
\tfrac{\partial}{\partial t}M(t;\lambda)&=\tfrac{\partial}{\partial t}\bra{}e^{\lambda \hat{n}}\ket{\Psi(t)}\\
&=\bra{}\left[\beta(e^{\lambda}-1) a^{\dag}+\tau(e^{-\lambda}-1) a\right]e^{\lambda \hat{n}}\ket{\Psi(t)}\\
&=\bra{}\left[\beta(e^{\lambda}-1)+\tau(e^{-\lambda}-1) a^{\dag}a\right]e^{\lambda \hat{n}}\ket{\Psi(t)}\\
&=\left(\beta(e^{\lambda}-1)+ \tau(e^{-\lambda}-1)\tfrac{\partial}{\partial \lambda}\right)M(t;\lambda)\,.
\end{aligned}
\end{equation}
In the second to last step, we invoked the auxiliary formula $\bra{}a^{\dag\:k}a^{\ell}=\bra{}a^{\dag\:\ell}a^{\ell}$, which follows from combining~\eqref{eq:HWcanrepClassical} with the definition~\eqref{eq:projBra} of the dual projection vector $\bra{}$. In combination with~\eqref{eq:discrStirling}, this yields the following useful identity:
\begin{equation}\label{eq:JCdiscr}
    \bra{}a^{\dag\:k}a^{\ell}\: e^{\lambda \hat{n}}=\bra{}a^{\dag\:\ell}a^{\ell}\: e^{\lambda \hat{n}}=\sum_{r=0}^{\ell}s_1(\ell;r)\bra{}\hat{n}^r\:e^{\lambda \hat{n}}
    =\sum_{r=0}^{\ell}s_1(\ell;r)\tfrac{\partial^r}{\partial\lambda^r}\bra{}e^{\lambda \hat{n}}\,.
\end{equation}
In Section~\ref{sec:DB}, it will be this type of identity that will provide the intuition for the generalized concept of \emph{polynomial jump-closure} necessary in the setting of generic rewriting systems. Finally, choosing for sake of illustration as an initial state a pure state with $N$ particles, $\ket{\Psi(0)}=\ket{N}$ (whence $M(0;\lambda)=e^{\lambda N}$), we may once again employ the technique of semi-linear normal ordering to compute the closed-form solution for $M(t;\lambda)$:
\begin{equation*}
M(t;\lambda)=e^{\tfrac{\beta}{\tau}(1-e^{-\tau t})(e^{\lambda}-1)}\cdot
\left(1+e^{-\tau t}(e^{\lambda}-1)\right)^N\,.
\end{equation*}
Upon comparison to the result~\eqref{eq:solBDP} computed in Delbr\"uck's formalism, we may thus verify the well-known relation $M(t;\lambda)=P(t;e^{\lambda})$.
\end{exa}

The ``stochastic mechanics'' framework as presented here for describing statistical moments of chemical reaction systems in the form of exponential moment generating functions will be generalized to the case of rewriting systems over arbitrary adhesive, extensive categories in the main part of this paper. Crucially, while the choice of methodology in terms of focusing on either the \emph{probability generating functions} or on the \emph{exponential moment generating functions} is fully equivalent in the case of chemical reaction systems, the latter is the only generic choice possible in the case of general stochastic rewriting systems. These aspects will be discussed in further detail in Sections~\ref{sec:egfDynamics} and~\ref{sec:CC}. As advertised in the introduction, chemical reaction systems will be identified (cf.\ Example~\ref{ex:discrGRS}) as the special case of the rewriting of \emph{discrete} (i.e. edge-less) graphs. This discovery in turn will play an important role in the proof of Theorem~\ref{thm:DMB} (establishing a somewhat surprising relationship between certain general rewriting systems and chemical reaction systems).

\section{Stochastic mechanics framework for DPO rewriting over adhesive, extensive categories}
\label{sec:SMF}

For the readers' convenience, we will briefly recall the essential concepts of the \emph{rule algebra formalism} for the special case of \emph{Double Pushout (DPO) rewriting}, in the recently developed variant for DPO rewriting over arbitrary adhesive, extensive\footnote{Referring to~\cite{Lack:2004aa} for the precise details, recall that an adhesive, extensive category is a category that is adhesive and that possesses a strict initial object. Note however that not every extensive category is adhesive.} categories as introduced in~\cite{Behr2018} (compare to~\cite{bdg2016,bdgh2016} for the original version based on relational concepts). 

In the traditional approach to DPO rewriting (see e.g.~\cite{CorradiniMREHL97,DBLP:conf/gg/1997handbook}), one of the core definitions is that of the action of a rewriting rule on an object of the underlying adhesive, extensive category via a particular match. Non-determinism arises thus due to the in general multiple possible matches in applying a given rewriting rule. For the special case of \emph{discrete graph rewriting} (or rewriting of finite sets), we observed in~\cite{bdg2016} that this type of non-determinism gives rise to the famous \emph{Heisenberg-Weyl algebra} and its canonical representation~\cite{blasiak2010combinatorial,blasiak2011combinatorial}. Thus this famous algebra is a special case of a so-called \emph{rule algebra}, which is the associative unital algebra of (concurrent) compositions of rewriting rules. Since the main focus of this paper consists in certain applications and extensions of this formalism and not the definition of the formalism itself, suffice it here for brevity to recall the basic definitions (compare Figure~\ref{fig:sketch}), referring the readers to~\cite{Behr2018} for the full technical details.

\begin{figure}[t]
  \centering
 \includegraphics[width=0.9\textwidth,page=1]{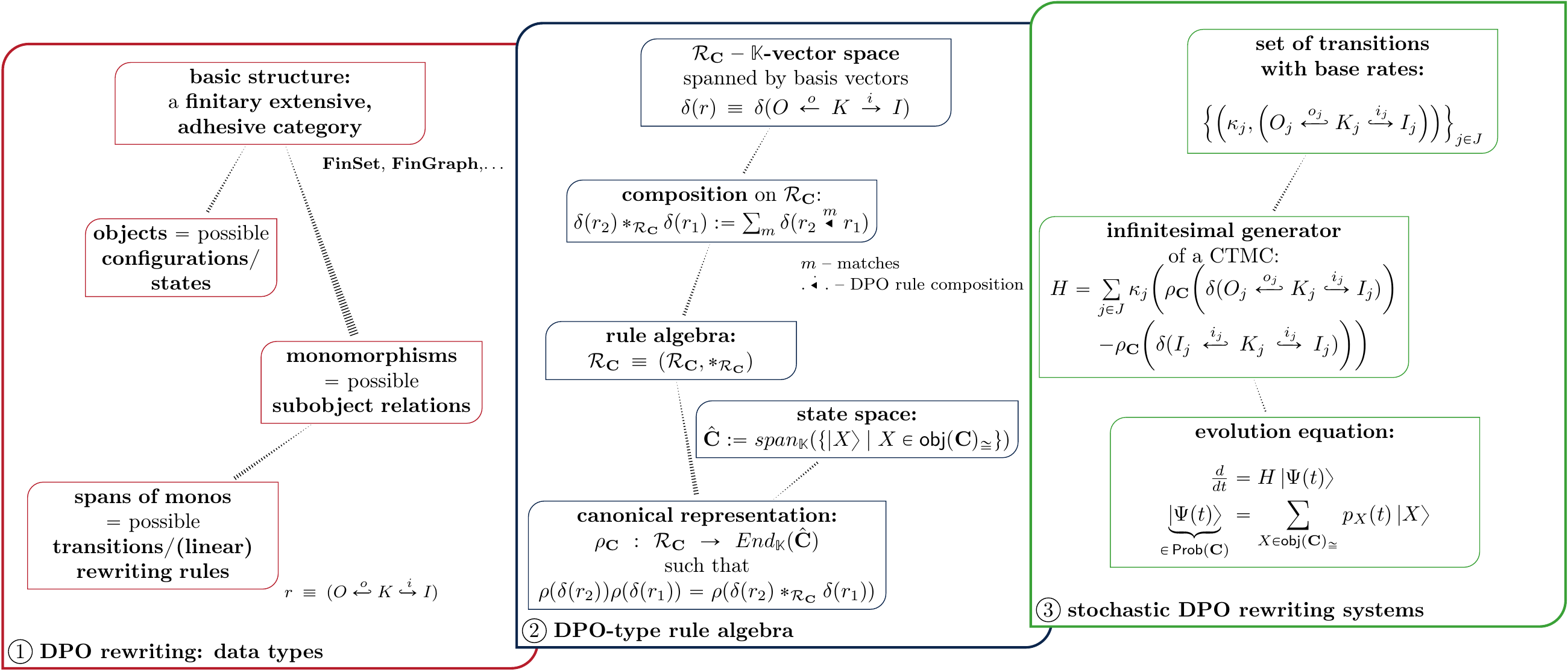}
  \caption{Sketch of the DPO-type rule algebra framework.}
  \label{fig:sketch}
\end{figure}

\subsection{Core definitions of DPO rewriting}

In the remainder of this paper, we will assume that all of our rewriting systems are formulated over \emph{adhesive, extensive categories}~\cite{Lack:2004aa,Lack:2005aa} $\bfC$ of \emph{finitary} type~\cite{GABRIEL_2014} (such as $\mathbf{FinSet}$ (finite sets), $\mathbf{FinGraph}$ (finite directed multigraphs), finite hypergraph categories etc.; see e.g.\ \cite{Lack:2004aa,Lack:2005aa} for further details). A category $\bfC$ being adhesive and extensive entails (amongst other properties) that
\begin{itemize}[label=$\triangleright$]
\item $\bfC$ possesses a \emph{strict initial object}, customarily denoted $\emptyset\in \obj{\bfC}$, such that for each object $X\in \obj{\bfC}$ there exists precisely one monomorphism $\emptyset\hookrightarrow X$,
\item pullbacks and pushouts along monomorphisms exist, and
\item pushout complements are unique.
\end{itemize}

The key concept in DPO rewriting (and from which the moniker originates) is the following one, traditionally referred to as the notion of \emph{direct derivations}:
\begin{defi}\label{def:DPOR}
	Let $\bfC$ be a finitary adhesive, extensive category. Define $\Lin{\bfC}$, the \emph{set of linear rules}, as the set of isomorphism classes of spans\footnote{Contrary to most of the standard rewriting literature, we opt to interpret spans of monomorphisms to encode partial injective morphisms from the \emph{right to the left leg} of the span, rather than in the opposite orientation; this will prove advantageous when considering operator compositions, as we may then use a convention in line with most of the mathematics literature.} of monomorphisms,
\begin{equation*}
	\Lin{\bfC}:=\{
		\rSpan{O}{o}{K}{i}{I}\mid O,K,I\in \obj{\bfC},\; o,i\in \mono{\bfC}
	\}\diagup_{\cong}\,.
\end{equation*} 
Here\footnote{The notation $A\equiv B$ will signify that ``$A$ is an equivalent or shorthand notation for $B$'' (assuming information explicit in $B$ but not in $A$ is clear from the context).}, two spans $r_j\equiv \rSpan{O_j}{o_j}{K_j}{i_j}{I_j}\in \Lin{\bfC}$ (for $j=1,2$) are isomorphic iff there exist three isomorphisms $\omega:O_1\rightarrow O_2$, $\kappa:K_1\rightarrow K_2$ and $\iota:I_1\rightarrow I_2$ such that the diagram
\[
\begin{mycd}
	O_1
		\ar[d,"\cong","\omega"'] & 
	K_1
		\ar[l,hook',"o_1"']
		\ar[r,"i_1"]
		\ar[d,"\cong","\kappa"'] & 
	I_1
		\ar[d,"\cong","\iota"']\\
	O_2 & 
	K_2 
		\ar[l,hook',"o_2"]
		\ar[r,hook,"i_2"'] &
	I_2
\end{mycd}
\]
commutes.

Given a linear rule $r\equiv \rSpan{O}{o}{K}{i}{I}\in \Lin{\bfC}$, an object $X\in \obj{\bfC}$ and a monomorphism $m:I\hookrightarrow X$, if the diagram below is constructable (i.e.\ if the pushout complement marked $POC$ exists),
\begin{equation}\label{eq:DPOdd}
\begin{mycd}
	O
		\ar[d,hook,dotted]
		\ar[dr,phantom,"PO"] & 
	K
		\ar[l,hook',"o"']
		\ar[r,"i"]
		\ar[d,hook,dashed] & 
	I
		\ar[d,hook,"m"]
		\ar[dl,phantom,"POC"]\\
	r_m(X) & 
	K' 
		\ar[l,hook',dotted]
		\ar[r,hook,dashed] &
	X
\end{mycd}
\end{equation}
we refer to the monomorphism $m:I\hookrightarrow X$ as an \emph{admissible match} of rule $r$ in $X$, while $r_m(X)$ denotes the object resulting from this application. For each linear rule $r$ and object $X$, we denote the \emph{set of admissible matches} by $\MatchSet{r}{X}$. More precisely, since the operations of pushout and pushout complement are only unique up to isomorphism, we must understand the above definitions as defined up to isomorphisms, i.e.\ $r_m(X)$ is understood as the isomorphism class whose representatives are calculated by picking representatives for $r$, $X$ and $m$ as well as for the relevant pushouts, and correspondingly (in a slight abuse of notations) the set $\MatchSet{r}{X}$ is also being defined up to isomorphisms of $r$ and $X$.
\end{defi}

By standard practice and for notational convenience, we will often keep the dependence on the various notions of isomorphisms implicit, such as in the ensuing lemma (where e.g.\ ``$r_{\emptyset_{m_X}}(X)=X$'' should be understood as ``$r_{\emptyset_{m_X}}(X)\cong X$'').

\begin{lem}\label{lem:trivRule}
	For every adhesive, extensive category $\bfC$, there exists a special linear rule 
	\[
		r_{\emptyset}\equiv \rSpanAlt{\emptyset}{}{\emptyset}{}{\emptyset}\in \Lin{\bfC}\,,
	\]
	referred to henceforth as the \emph{trivial rule}, with the property that for every object $X\in \obj{\bfC}$ there exists precisely one admissible match $m_X:\emptyset\hookrightarrow X \in \MatchSet{r_{\emptyset}}{X}$, with 
\begin{equation}\label{eq:DPOddTriv}
\begin{mycd}
	\emptyset
		\ar[d,hook,dotted]
		\ar[dr,phantom,"PO"] &[-15pt] 
	\emptyset
		\ar[l,hook']
		\ar[r,hook]
		\ar[d,hook,dashed] &[10pt]
	\emptyset
		\ar[d,hook,"m_X"]
		\ar[dl,phantom,"POC"]\\
	r_{\emptyset_{m_X}}(X)=X & 
	X 
		\ar[l,equal,dotted]
		\ar[r,equal,dashed] &
	X
\end{mycd}\,.
\end{equation}
\begin{proof}
	By virtue of the fact that $\emptyset\in \obj{\bfC}$ is the strict initial object of $\bfC$, there exists a unique monomorphism $m_X:\emptyset\hookrightarrow X$. It thus remains to prove that for a given object $X$, the morphism $m_X$ is indeed an admissible match. It follows from standard category-theoretical properties that for any monomorphism $(A\hookrightarrow B)\in \mono{\bfC}$, the following square is both a pushout and a pullback:
	\begin{equation*}
	\begin{mycd}
		A \ar[d,hook]\ar[r,equal] & A\ar[d,hook]\\
		B \ar[r,equal] & B
	\end{mycd}\,.
	\end{equation*}
\end{proof}
\end{lem}

The second key ingredient (in particular in view of the construction of rule algebras, albeit with slightly less focus in the traditional rewriting literature) concerns the sequential compositions of linear rules. Intuitively, upon applying two linear rules in sequence to an object of the underlying category according to Definition~\ref{def:DPOR}, it is natural to imagine that one might alternatively first ``precompute'' a form of sequential composition of the rules themselves prior to application to the object. This concept will prove extremely versatile in practical computations in stochastic rewriting systems (see the later sections in this paper). 

\begin{defi}[Sequential compositions of linear rules]\label{def:seqCompDPOR}
	Let $r_j\equiv \rSpan{O_j}{o_j}{K_j}{i_j}{I_j}\in \Lin{\bfC}$ (for $j=1,2$) be two linear rules. Then a span of monomorphisms $\mu\equiv \rSpanAlt{I_1}{}{M}{}{O_2}$ constitutes an \emph{admissible match of $r_1$ into $r_2$} if and only if the following diagram is constructable (whence if the pushout complements marked $POC$ exist):
	\begin{equation}\label{eq:RcompSeq}
	\begin{mycd}
		&
		&
		&
		K_{12}
			\ar[dll,hook,dotted]
			\ar[dlll,hook',bend right=10,"o_{12}"']
			\ar[dlll,phantom,"="]
			\ar[drr,hook',dotted]
			\ar[drrr,hook,bend left=10,"i_{12}"]
			\ar[drrr,phantom,"="]
			\ar[d,phantom,"PB"] &
		&
		&\\
		O_{12} &
		K_{1}'
			\ar[l,hook',dotted]
			\ar[rr,hook,dashed]
			\ar[dr,phantom,"POC"] 
			\ar[dl,phantom,"PO"] &
		&
		M' &
		&
		K_{2}' 
			\ar[r,hook,dotted]
			\ar[ll,hook',dashed]
			\ar[dl,phantom,"POC"]
			\ar[dr,phantom,"PO"] &
		I_{12}\\
		O_1 
			\ar[u,hook,dotted] &
		K_1
			\ar[l,hook',"o_1"]
			\ar[r,hook,"i_1"']
			\ar[u,hook,dashed] &
		I_1 
			\ar[ur,hook,bend left] &
		M 
			\ar[l,hook']
			\ar[r,hook]
			\ar[u,phantom,"PO"] & 
		O_2 
			\ar[ul,hook',bend right] &
		K_2 
			\ar[l,hook',"o_2"]
			\ar[r,hook,"i_2"']
			\ar[u,hook,dashed] &
		I_2
			\ar[u,hook,dotted]
	\end{mycd}
	\end{equation}
Here, the squares marked $PO$ are constructed by taking pushouts, the one marked $PB$ by taking pullback, and the triangles marked $=$ are commuting triangles of morphisms. In case the diagram is constructable, the resulting linear rule
\begin{equation*}
	r_{12}\equiv \rSpan{O_{12}}{o_{12}}{K_{12}}{i_{12}}{I_{12}}\in \Lin{\bfC}
\end{equation*}
is called the \emph{composite of $r_1$ with $r_2$ along the match $\mu$}, which we also denote by
\begin{equation*}
	r_{12}=\comp{r_1}{\mu}{r_2}\,.
\end{equation*}
The \emph{set of admissible matches of $r_1$ into $r_2$} is denoted $\MatchSet{r_1}{r_2}$. More precisely, both $\MatchSet{r_1}{r_2}$ and $r_{12}$ are defined up to isomorphisms (where isomorphisms of matches are induced by isomorphisms of the linear rules).
\end{defi}

\begin{lem}[Trivial match of linear rules]\label{lem:rTrivMatch}
	For any two linear rules ($j=1,2$)
	\[
	r_j\equiv\rSpan{O_j}{o_j}{K_j}{i_j}{I_j}\in \Lin{\bfC}\,,
	\]
	the set of admissible matches $\MatchSet{r_1}{r_2}$ is non-empty, since it at least contains the \emph{trivial match} $\mu_{\emptyset}\equiv \rSpanAlt{I_1}{}{\emptyset}{}{O_2}$. Moreover, the composition along trivial matches is commutative, in the sense that if $\mu_{\emptyset}'\in \MatchSet{r_2}{r_1}$ is the trivial match of $r_2$ into $r_1$,
	\begin{equation*}
		\comp{r_1}{\mu_{\emptyset}}{r_2}=\comp{r_2}{\mu_{\emptyset}'}{r_1}\,.
	\end{equation*}
\begin{proof}
	Admissibility of the trivial match $\mu_{\emptyset}\in \MatchSet{r_1}{r_2}$ for two linear rules $r_1$ and $r_2$ follows from the construction of the following diagram (where $+$ denotes the operation of disjoint union):
	\begin{equation}\label{eq:RcompSeqTrivMatch}
	\begin{mycd}
		&
		&
		&
		K_1 +  K_2
			\ar[dll,hook,dotted]
			\ar[dlll,hook',bend right=10,"o_1 +  o_2"']
			\ar[dlll,phantom,"="]
			\ar[drr,hook',dotted]
			\ar[drrr,hook,bend left=10,"i_1 +  i_2"]
			\ar[drrr,phantom,"="]
			\ar[d,phantom,"PB"] &
		&
		&\\
		O_1 +  O_2 &
		K_1 +  O_2
			\ar[l,hook',dotted]
			\ar[rr,hook,dashed]
			\ar[dr,phantom,"POC"] 
			\ar[dl,phantom,"PO"] &
		&
		I_1 +  O_2 &
		&
		I_1 +  K_2 
			\ar[r,hook,dotted]
			\ar[ll,hook',dashed]
			\ar[dl,phantom,"POC"]
			\ar[dr,phantom,"PO"] &
		I_1 +  I_2\\
		O_1 
			\ar[u,hook,dotted] &
		K_1
			\ar[l,hook',"o_1"]
			\ar[r,hook,"i_1"']
			\ar[u,hook,dashed] &
		I_1 
			\ar[ur,hook,bend left] &
		\emptyset 
			\ar[l,hook']
			\ar[r,hook]
			\ar[u,phantom,"PO"] & 
		O_2 
			\ar[ul,hook',bend right] &
		K_2 
			\ar[l,hook',"o_2"]
			\ar[r,hook,"i_2"']
			\ar[u,hook,dashed] &
		I_2
			\ar[u,hook,dotted]
	\end{mycd}
	\end{equation}
	Here, the fact that this diagram is always constructable hinges on a property intrinsic to adhesive categories known as ``pushout-pullback decomposition'' (cf.\ e.g. \cite{Lack:2004aa}, Lemma~29): in the case at hand, for any diagram of the form
	\begin{equation*}
\begin{mycd}
	C  \ar[r,hook] \ar[rr, bend left, hook]  & 
	A +  C \ar[r,hook] & 
	B +  C \\
	\emptyset \ar[r,hook] \ar[rr,bend right, hook] \ar[u,hook] & 
	A \ar[r,hook] \ar[u,hook]  & B \ar[u,hook]
\end{mycd}\,,
\end{equation*}
the outer square is a pushout, the right square a pullback, and thus by pushout-pullback decomposition all squares are both pushouts and pullbacks. This proves the existence of the relevant pushout squares marked $POC$ in~\eqref{eq:RcompSeqTrivMatch}. Finally, commutativity of the composition of linear rules along trivial matches follows from the symmetry of the diagram in~\eqref{eq:RcompSeqTrivMatch}.
\end{proof}
\end{lem}

As a side remark, we note that for two linear rules $r_1,r_2\in \Lin{\bfC}$, their composition along the trivial match results in a linear rule that in the context of traditional rewriting would be interpreted as a (sequentially independent) parallel rule.

\subsection{From DPO rewriting to DPO rule algebras}
\label{sec:DPOrToRAs}

The definition of composition of linear rules along admissible matches entails that in general there might exist multiple possible matches for two given rules. One possibility to ``encode'' this particular form of non-determinism consists in the following notion of so-called \emph{rule algebras}:

\begin{defi}[DPO rule algebras; cf.~\cite{Behr2018}, Def.~4.2]\label{def:DPOra}
	Let $\bfC$ be a finitary adhesive, extensive category, $\Lin{\bfC}$ the set of linear rules over $\bfC$, and $\comp{.}{.}{.}$ the composition operation introduced in Definition~\ref{def:seqCompDPOR}. Let $\bK$ be a field (of characteristic $0$, e.g.\ $\bK=\bR$), and define $\cR_{\bfC}$ as the \emph{$\bK$-vector space with basis vectors $\delta(r)$ indexed by elements $r\in\Lin{\bfC}$}:
	\begin{equation*}
      \cR_{\bfC}:=span_{\bK}\left(\left\{
      	\delta(r)
      \right\}_{r\in \Lin{\bfC}}\right)\,.
	\end{equation*}
	We equip the $\bK$-vector space $\cR_{\bfC}$ with a \emph{bilinear multiplication} $\ast_{\cR_{\bfC}}$ via defining\footnote{Note that according to Lemma~\ref{lem:rTrivMatch}, the set of admissible matches $\MatchSet{r_1}{r_2}$ for arbitrary linear rules $r_1,r_2\in \Lin{\bfC}$ is non-empty, and moreover finite due to finitarity of $\bfC$, whence the definition of $\ast_{\cR_{\bfC}}$ as provided in~\eqref{eq:defRAM} is well-posed.}
	\begin{equation}\label{eq:defRAM}
	\begin{aligned}
	\ast_{\cR_{\bfC}}&:\cR_{\bfC}\times \cR_{\bfC}\rightarrow \cR_{\bfC}:\\
	&\quad (\delta(r_1),\delta(r_2))\mapsto \delta(r_1)\ast_{\cR_{\bfC}}\delta(r_2):=\sum_{\mu\in \MatchSet{r_1}{r_2}}\delta\left(\comp{r_1}{\mu}{r_2}\right)\,,
	\end{aligned}
	\end{equation}
	extended to arbitrary elements of $\cR_{\bfC}$ by linearity. We refer to
	\[
		\cR_{\bfC}\equiv (\cR_{\bfC},*_{\cR_{\bfC}})
	\]
	as the \emph{DPO rule algebra over $\bfC$}.
\end{defi}

We cite from~\cite{Behr2018} the following important result:
\begin{thm}[\cite{Behr2018}, Thm.~4.3]
	For every finitary adhesive, extensive category $\bfC$, the DPO rule algebra $\cR_{\bfC}$ as introduced in Definition~\ref{def:DPOra} is a \emph{unital associative algebra}, with unit
	\[
	R_{\emptyset}:= \delta(r_{\emptyset})\,,\quad r_{\emptyset}:= \rSpanAlt{\emptyset}{}{\emptyset}{}{\emptyset}\in \Lin{\bfC}
	\]
	for the composition operation $\ast_{\cR_{\bfC}}$.
\begin{proof}
	The fact that $R_{\emptyset}$ constitutes the unit element for $\ast_{\cR_{\bfC}}$ follows directly from the definition of $\ast_{\cR_{\bfC}}$ as given in~\eqref{eq:defRAM} and by specializing Lemma~\ref{lem:rTrivMatch} to the cases $r_1=r_{\emptyset}$ and $r_2=r_{\emptyset}$ (thus verifying that $R_{\emptyset}$ is both a left and a right unit for $\ast_{\cR_{\bfC}}$). For the somewhat involved proof of the associativity property, we refer the interested readers to~\cite{Behr2018}.
\end{proof}
\end{thm}

As a first important example for a calculation involving the concept of a rule algebra, consider the following special case (which will provide a running example due to its importance for explaining the relationship between chemical reactions and stochastic rewriting systems):
\begin{exa}\label{exa:HWra1}
    Consider the following special subclass of linear rewriting rules in $\mathbf{FinGraph}$:
    \begin{equation*}
        r_{p,q}:=\rSpanAlt{\bullet^{\uplus\:p}}{}{\emptyset}{}{\bullet^{\uplus\:p}}\,.
    \end{equation*}
    Here, for each $n\in\bZ_{\geq0}$, $\bullet^{\uplus\:n}$ denotes the (isomorphism class of the) \emph{discrete graph with $n$ vertices}. It is an easy exercise to compute (with $*\equiv *_{\cR_{\mathbf{FinGraph}}}$)
    \begin{equation}\label{eq:hwA}
        \delta(r_{p_2,q_2})*\delta(r_{p_1,q_1})=\sum_{k=0}^{min(q_2,p_1)}k!\binom{q_2}{k}\binom{p_1}{k}\delta(r_{p_1+p_2-k,q_1+q_2-k})\,.
    \end{equation}
    In fact, the coefficient of the $k$-th summand expresses precisely the number of possibilities to pair $k$ vertices of the \emph{input interface} $\bullet^{\uplus\:q_2}$ of the rule $r_{p_2,q_2}$ with $k$ vertices of the \emph{output interface} $\bullet^{\uplus\:p_1}$ of the rule $r_{p_1.q_1}$ while disregarding the order of forming the pairs. Let us introduce a special notation for two particular rule algebra elements\footnote{The notation $x^{\dag}$ is chosen due to the fact that $r_{1,0}=r_{0,1}^{\dag}$, in the sense that the span of monomorphisms $r_{1,0}$ is the reverse of the span $r_{0,1}$.} of the form $\delta(r_{p,q})$:
    \begin{equation*}
        x^{\dag}:=\delta(r_{1,0})=\delta\rSpanAlt{\bullet}{}{\emptyset}{}{\emptyset}\,,\; x:=\delta(r_{0,1})=\delta\rSpanAlt{\emptyset}{}{\emptyset}{}{\bullet}
    \end{equation*}
    Utilizing the formula~\eqref{eq:hwA} repeatedly, one may derive the following set of important results~\cite{Behr2018}:
    \begin{subequations}
    \begin{align}
        x^{\dag\:*p}*x^{\:*q}&:=\overset{\text{$p$ times}}{x^{\dag}*\dotsc x^{\dag}}*\overset{\text{$q$ times}}{x^{\vphantom{\dag}}*\dotsc x}=\delta(r_{p,q})\label{eq:raHW1}\\
        [x,x^{\dag}]&=x*x^{\dag}-x^{\dag}*x=\delta(r_{0,0})=\delta\rSpanAlt{\emptyset}{}{\emptyset}{}{\emptyset}\label{eq:raHW2}\,.
    \end{align}
    \end{subequations}
    Equation~\eqref{eq:raHW1} has a simple explanation from the viewpoint of DPO rewriting: in this particular order of composition ($q$ vertex deletion rules followed by $p$ vertex creation rules), one does not find any other admissible matches of rules in each of the rule compositions than the trivial match, whence the formula directly follows from the calculation presented in Lemma~\ref{lem:rTrivMatch}. Combining~\eqref{eq:hwA} and~\eqref{eq:raHW1}, the rule algebra elements $x^{\dag}$, $x$ and $R_{\emptyset}=\delta(r_{0,0})$ thus \emph{generate} a subalgebra of $\cR_{\mathbf{FinGraph}}$ with basis elements $\delta(r_{p,q})$ (in the sense that $\delta(r_{p,q})=x^{\dag\:*p}*x^{\:*q}$). Finally, the \emph{canonical commutation relation}~\eqref{eq:raHW2} expresses the fact that $x^{\dag}$ and $x$ do not commute.
\end{exa}

The fact of $\cR_{\bfC}$ being a unital associative algebra entails that one may define a natural notion of \emph{representation} for $\cR_{\bfC}$, which finally re-establishes the contact of the rule algebra formalism with the traditional notion of DPO rewriting:

\begin{defi}[Canonical representations; cf.~\cite{Behr2018}, Def.~4.4]\label{def:canRep}
	For a finitary adhesive, extensive category $\bfC$, let $\cR_{\bfC}$ be the associated DPO rule algebra. Define $\hat{\bfC}$ as the \emph{free $\bK$-vector space with basis vectors $\ket{X}$ indexed by isomorphism classes of objects of $\bfC$},
\begin{equation*}
	\hat{\bfC}:=span_{\bK}\left(\left\{\ket{X}\right\}_{X\in \obj{\bfC}_{\cong}}\right)\,,
\end{equation*}
and denote by $End_{\bK}(\hat{\bfC})$ the space of $\bK$-linear endomorphisms of $\hat{\bfC}$. Then the \emph{canonical representation} $\rho_{\bfC}\equiv \rho_{\cR_{\bfC}}$ is defined as
\begin{equation}\label{eq:defCanRep}
\begin{aligned}
	\rho_{\bfC}&: \cR_{\bfC}\rightarrow End_{\bK}(\hat{\bfC}):\\
	&\quad \forall r\in \cR_{\bfC}, X\in \obj{\bfC}:\quad \rho_{\bfC}(\delta(r))\ket{X}:=\sum_{m\in \MatchSet{r}{X}} \ket{r_m(X)}\,,
\end{aligned}
\end{equation}
extended to arbitrary elements of $\cR_{\bfC}$ and of $\hat{\bfC}$ by linearity (and with notations $\MatchSet{r}{X}$ and $r_m(X)$ as introduced in Definition~\ref{def:DPOR}). 
\end{defi}

\begin{thm}[cf.~\cite{Behr2018}, Thm.~4.5]\label{thm:canrep}
	The canonical representation $\rho_{\bfC}$ of the DPO rule algebra $\cR_{\bfC}$ over some finitary adhesive, extensive category $\bfC$ is a \emph{homomorphism of unital associative algebras}, and thus indeed constitutes a well-posed representation of $\cR_{\bfC}$.
\begin{proof}
	The claim entails that $\rho_{\bfC}$ maps the unit element of $\cR_{\bfC}$ to the unit element of $End_{\bK}(\hat{\bfC})$, and that for all $r_1,r_2\in \Lin{\bfC}$,
	\[
		\rho_{\bfC}(\delta(r_1))\rho_{\bfC}(\delta(r_2))=\rho_{\bfC}(\delta(r_1\ast_{\cR_{\bfC}}r_2))\,.
	\]
	The first part of the claim follows directly from Lemma~\ref{lem:trivRule} and from the Definition of $\rho_{\bfC}$ as specified in~\eqref{eq:defCanRep}. We refer the interested readers to~\cite{Behr2018} for the proof of the second part of the claim (which involves a variant of the so-called concurrency theorem from the traditional rewriting literature).
\end{proof}
\end{thm}

\begin{exa}[Ex.~\ref{exa:HWra1} continued]\label{exa:HWra2}
    Specializing the rule algebra representation $\rho=\rho_{\cR_{\mathbf{FinGraph}}}$ to the case of the subalgebra of $\cR_{\mathbf{FinGraph}}$ generated by $x^{\dag}$, $x$ and $R_{\emptyset}$ as defined in Example~\ref{exa:HWra1}, it is straightforward to verify that (with $\mathbb{1}:=\mathbb{1}_{\mathbf{FinGraph}}$)
    \begin{equation*}
        \rho(R_{\varnothing})\ket{X}=\mathbb{1}\ket{X}\,.
    \end{equation*}
    In order to provide an interpretation for the linear operators $a^{\dag}:=\rho(x^{\dag})$ and $a:=\rho(x)$, let us introduce the notation
    \begin{equation*}
        \ket{n}:=\ket{\bullet^{\uplus\:n}}\qquad (n\in \bZ_{\geq0})
    \end{equation*} 
    for the basis vector associated to the isomorphism class of a \emph{discrete} (i.e.\ edgeless) graph with $n$ vertices (with $\ket{0}:=\ket{\emptyset}$). Invoking the definition of $\rho$ according to~\eqref{eq:defCanRep}, we find the following identities:
    \begin{equation}\label{eq:HWcanrepRA}
        a^{\dag}\ket{n}=\ket{n+1}\,,\; a\ket{0}=0\,,\;\forall n>0:\; a\ket{n}=n\ket{n-1}\,.
    \end{equation}
    Note that these formulae precisely coincide with those in~\eqref{eq:HWcanrepClassical}. Finally, we illustrate the typical usage scenario for the representation property described in the proof of Theorem~\ref{thm:canrep} in order to map the commutation relation~\eqref{eq:raHW2} of Example~\ref{exa:HWra1} into the following statement:
    \begin{equation*}
    \begin{aligned}
        \rho\left([x,x^{\dag}]\right)&=\rho(x*x^{\dag}-x^{\dag}*x)=\rho(x)\rho(x^{\dag})-\rho(x^{\dag})\rho(x)=aa^{\dag}-a^{\dag}a=\mathbb{1}\,.
    \end{aligned}
    \end{equation*}
    Since this is precisely the \emph{canonical commutation relation}~\eqref{eq;HWccr}, we find that on the subspace of $\widehat{\mathbf{FinGraph}}$ spanned by the discrete graph states $\ket{n}$, the operators $a^{\dag}$ and $a$ coincide with the classical generators of the so-called \emph{Heisenberg-Weyl algebra}~\cite{Behr2018}. However, in contrast to the classical formulation, our new formulation permits to also act on \emph{non-discrete} graph states. For example, one finds that
    \begin{equation}\label{eq:raNOF}
        a^{\dag}\ket{\TwoVertDirEdgeG[]}=\ket{\OneVertG[]\;\TwoVertDirEdgeG[]}\,,\;
        a\ket{\TwoVertDirEdgeG[]}=0\,.
    \end{equation}
    In other words, $a^{\dag}$ implements the addition of a disjoint vertex to a graph, while $a$ implements the deletion of a vertex, with the second equation above expressing the fact that it is not possible in DPO semantics to delete vertices with incident edges (unless the incident edges are explicitly deleted as well via the choice of a suitable rule). Finally, utilizing the results of Example~\ref{exa:HWra1}, we may derive the important identity
    \begin{equation*}
        a^{\dag\:p}a^{\:q}=\rho(x^{\dag\:*p}*x^{\:*q})=\rho\left(\delta\rSpanAlt{\bullet^{\uplus\:p}}{}{\emptyset}{}{\bullet^{\uplus\:q}}\right)\,.
    \end{equation*}
    This identity provides a one-to-one correspondence between $a^{\dag\:p}a^{\:q}$ (sometimes referred to as a term in \emph{normal-ordered form}, i.e.\ with all occurrences of $a^{\dag}$ to the left) and the representation of the rule algebra element $\delta(r_{p,q})$.
\end{exa}

For later convenience, we introduce the following (commutative) operation motivated from the results of Lemma~\ref{lem:rTrivMatch}:
\begin{defi}\label{def:sup}
	Let $r_1,r_2\in \Lin{\bfC}$ be two linear rules. Then we define the \emph{induced superposition operation} $\squplus$ as
	\begin{equation*}
		\rho_{\bfC}(\delta(r_1))\squplus \rho_{\bfC}(\delta(r_2)):=\rho_{\bfC}\left(\delta\left(
		\comp{r_1}{\mu_{\emptyset}}{r_2}
		\right)\right)\,,
	\end{equation*}
with $\mu_{\emptyset}\in \MatchSet{r_1}{r_2}$ denoting the trivial match of $r_1$ into $r_2$.
\end{defi}
It is important to note that the operation $\squplus$ thus defined does \emph{not} coincide with the operation $\oplus$ of superposition of linear operators, which would entail a splitting of the basis of the underlying vector space into two disjoint invariant subspaces. Instead, the operation $\squplus$ lifts the sequentially independent application of linear rules to the level of representations.\\

Finally, we note the following property of canonical representations, which poses an important prerequisite for constructing stochastic rewriting systems:

\begin{lem}[compare \cite{bdg2016}, Lemma~30 and \cite{Behr2018}, proof of Thm.~7.3]
\label{lem:crFin}
For every finitary adhesive, extensive category $\bfC$, the canonical representation $\rho_{\bfC}$ of the DPO rule algebra $\cR_{\bfC}$ over $\bfC$ ranges in \emph{column- and row-finite} linear operators on $\hat{\bfC}$, whence: 
\begin{itemize}[label=$\triangleright$]
\item For an arbitrary rule algebra basis element $R\equiv\delta(r)\in \cR_{\bfC}$ (with $r\in \Lin{\bfC}$ a linear rule) and an arbitrary pure state $\ket{X}\in \hat{\bfC}$, $\rho_{\bfC}(R)\ket{X}$ yields a \emph{finite} linear combination of states.
\item Moreover, for a given pure state $\ket{Y}\in \hat{\bfC}$ and for a fixed element $R\equiv\delta(r)\in \cR_{\bfC}$, there exist only finitely many pure states $\ket{X}\in \hat{\bfC}$ such that $\ket{Y}$ is a summand in $\rho_{\bfC}(R)\ket{X}$.
\end{itemize}
\begin{proof}
	The claim follows directly from the fact that for any given linear rule 
	\[
		r\equiv\rSpan{O}{o}{K}{i}{I}\in \Lin{\bfC}\,,
	\]
	there exist only finitely many monomorphisms $m:I\hookrightarrow X$ into any given object $X\in \obj{\bfC}$, and whence only finitely many admissible matches (cf.\ Definition~\ref{def:DPOR}). Conversely, the set of objects $X\in \obj{\bfC}$ that rewrite into a given object $Y\in \obj{\bfC}$ may be determined (by virtue of the symmetry involved in Definition~\ref{def:DPOR}) via computing the set of admissible matches for the application of the reverse rule $\bar{r}\equiv\rSpan{I}{i}{K}{o}{O}\in \Lin{\bfC}$, which is thus again a finite set.
\end{proof}
\end{lem}

\subsection{Stochastic mechanics for DPO rewriting systems}
\label{sec:SMDPOR}

Taking advantage of the DPO-type rule algebra framework, it is possible~\cite{bdg2016,Behr2018} to define the notion of continuous-time Markov chains (CTMCs) based on DPO rewriting rules in a fashion compatible with the general notions of stochastic mechanics~\cite{bdp2017}. We briefly recall the relevant material from~\cite{bdg2016} (see also~\cite{norris} for the general theory of CTMCs for comparison):
\begin{defi}
Let $\bfC$ be a finitary adhesive, extensive category whose set of isomorphism classes of objects $\obj{\bfC}_{\cong}:=\obj{\bfC}\diagup\cong$ is countable. The \emph{\textbf{state space}} $\cS_{\bfC}$ is defined as the Fr\'{e}chet space of real sequences $f=(f_X)_{X\in \obj{\bfC}_{\cong}}$ indexed by basis vectors $\ket{X}\in \hat{\bfC}$ (so-called \emph{pure states}) and with semi-norms $\|f\|_X:=|f_X|$,
\begin{equation*}
\cS_{\bfC}:=\left(\bR^{\hat{\bfC}},\left\{\|.\|_X\right\}_{X\in\obj{\bfC}_{\cong}}\right)\,.
\end{equation*}
The convex subset $\Prob(\bfC)\subset \cS_{\bfC}$ of \textbf{\emph{subprobability distributions over $\bfC$}}  is defined as
\begin{equation*}
\Prob(\bfC):=\left\{\left.\ket{\Psi}=\sum_{X\in \obj{\bfC}_{\cong}}\psi_X\ket{X}\in \cS_{\bfC}\right\vert \sum_{X\in \obj{\bfC}_{\cong}}\psi_X\leq 1\land \forall X\in \obj{\bfC}_{\cong}: \psi_X\in \bR_{\geq 0}\right\}\,.
\end{equation*}
The space of \emph{substochastic operators} $\Stoch(\cS_{\bfC})\subsetneq \End(\cS_{\bfC})$ is defined as the subset of the space of endomorphisms $\End(\cS_{\bfC})$ that leave $\Prob(\bfC)$ invariant. A linear operator $H\in \End(\cS_{\bfC})$ is called \emph{Hamiltonian} or \emph{infinitesimal stochastic operator} or conservative stable $Q$-matrix~\cite{wjanderson} if for $H\equiv(h_{X,Y})_{X,Y\in \obj{\bfC}_{\cong}}$,
\begin{equation}\label{eq:Hprops}
\begin{aligned}
(i)& & \forall X\in \obj{\bfC}_{\cong}&:\quad h_{X,X}\leq 0\\
(ii) && \forall X,Y\in \obj{\bfC}_{\cong}&: X\neq Y\Rightarrow h_{x,y}\geq 0\\
(iii)& & \forall Y\in \obj{\bfC}_{\cong}&: \sum_{X\in \obj{\bfC}_{\cong}}h_{X,Y}=0\,.
\end{aligned}
\end{equation}
Given a Hamiltonian $H\in \End(\cS_{\bfC})$, $H$ gives rise to an \emph{\textbf{evolution semi-group}} $\cE:[0,\infty)\rightarrow \Stoch(\cS_{\bfC})$ (i.e.\ $\cE(t)\cE(t')=\cE(t+t')$ for all $t,t'\in \bR_{\geq0}$) in the form of the point-wise minimal non-negative solution of the so-called \emph{Kolmogorov backwards equation}
\begin{equation*}
\frac{d}{dt}\cE(t)=H\cE(t)\,,\quad \cE(0)=\mathbb{1}_{\End(\cS_{\bfC})}\,.
\end{equation*}
The data of a Hamiltonian $H\in \End(\cS_{\bfC})$ and an \emph{initial state} $\ket{\Psi(0)}\in \Prob(\cS_{\bfC})$ determines the time-evolution of the corresponding \textbf{\emph{continuous-time Markov chain (CTMC)}} as~\cite{norris}
\begin{equation*}
\ket{\Psi(t)}=\cE(t)\ket{\Psi(0)}\,.
\end{equation*}
\end{defi}
While it appears to be an open problem of classifying those finitary adhesive, extensive categories that have the additional property that their class of isomorphism classes is a countable set, many categories of practical interest satisfy this property. Typical examples include $\mathbf{FinSet}$ (the category of finite sets and total set functions) and $\mathbf{FinGraph}$ (the category of finite directed multigraphs).\\ 

It then remains to determine how to construct Hamiltonians from rule algebra elements that encode linear DPO rewriting rules. It is important to note the following result from~\cite{bdg2016}:
\begin{lem}[compare \cite{bdg2016}, Lemmas~30 and~31 and \cite{Behr2018}, proof of Thm.~7.3]
The representation $\rho_{\bfC}(r)$ of an arbitrary rule algebra element $r\in \cR_{\bfC}$ furnishes an element of $\End(\cS_{\bfC})$.
\end{lem}
\begin{proof}
According to Lemma~\ref{lem:crFin}, $\rho_{\bfC}(r)$ is both a column- and a row-finite operator, from which the claim follows.
\end{proof}

\begin{defi}
The \emph{projection} $\bra{}$ is defined as the linear operation\footnote{We will consider $\bra{}$ throughout this paper as a partial linear map on $\cS_{\bfC}$.} 
\begin{equation*}
\bra{}: \ell^1_{\bR}(\hat{\bfC})\rightarrow \bR: \sum_{X\in \obj{\bfC}_{\cong}}\psi_X\ket{X}\mapsto \sum_{X\in \obj{\bfC}_{\cong}}\psi_X\,.
\end{equation*}
\end{defi}
\begin{lem}[compare~\cite{bdg2016}, Lemma~32]\label{lem:Haux}
For every Hamiltonian $H\in \End(\cS_{\bfC})$, $\bra{}H = 0$.
\begin{proof}
By definition of Hamiltonians, i.e.\ $\bra{}H\ket{X}=\sum_{Y\in \obj{\bfC}_{\cong}}h_{Y,X}=0_{\bR}$ for all $X\in \obj{\bfC}_{\cong}$.
\end{proof}
\end{lem}

The main objective in the study of stochastic rewriting systems consists in the following scenario: given a set of linear rules and of base frequencies\footnote{The SI unit of the rates is $[s^{-1}]$, yet since rates are always multiplied by the time parameter $t$ (whose SI unit is $[s]$), we will omit the units for brevity.} for each rule (also referred to as \emph{base rates}), one would like to determine the stochastic dynamics of \emph{observable quantities} of the system. For rewriting systems, such quantities are so-called \emph{observables}, whose precise implementation will be specified momentarily. According to the (minimal process) CTMC semantics, the system state is a time-dependent subprobability distribution over all possible configurations of the dynamical system. The quantities of interest are the (time-dependent) \emph{statistical moments} of the observables with respect to this time-dependent subprobability. (The CTMC semantics is in general a subprobability and not a probability because the chain can `explode' in finite time depending on $H$.)
\begin{prop}[\cite{bdg2016}, p.~7f and~\cite{Behr2018}, Thm.~7.3]\label{prop:Obs}
The space $\cO_{\bfC}\subset \End_{\bR}(\hat{\bfC})$ of linear operators that are \emph{diagonal} in the basis of pure states $\ket{X}\in \hat{\bfC}$ and that originate from finite linear combinations of operators that are representations of rule algebra elements is spanned by so-called \emph{(DPO)-type observables}, 
which are operators of the form
\begin{equation*}
O_P^k:=\rho_{\bfC}\left(\delta\left(o_P^k\right)\right)\,,
\quad o_P^k:=\rSpan{P}{k}{K}{k}{P}\in \Lin{\bfC}\,.
\end{equation*}
Their action on pure states $\ket{X}\in \hat{\bfC}$ is given by
\begin{equation*}
O_P^k\ket{X}=\omega_P^k(X)\ket{X}\,,\quad \omega_P^k(X)=|\MatchSet{o_P^k}{X}|\,,
\end{equation*}
with $\omega_M^k(X)\in \bZ_{\geq0}$ the \emph{number of admissible matches of the rule $o_P^k\in \Lin{\bfC}$ in $X\in \obj{\bfC}$}.
\begin{proof}
	The claim follows directly by observing that according to Definition~\ref{def:DPOR}, applying a linear rule of the form $o_P^k$ as above to an object $X$ in DPO rewriting, the diagram constructed according to~\eqref{eq:DPOdd} is \emph{symmetric}, which entails that for every admissible match $m:P\hookrightarrow X$, $(o_P^k)_m(X)=X$. 
\end{proof}
\end{prop}

\begin{exa}
Consider the finitary adhesive, extensive category $\mathbf{FinGraph}$ of \emph{finite directed multigraphs}. In order to illustrate the precise nature of observables as described in Proposition~\ref{prop:Obs}, we define the following three variants of \emph{edge-counting graph observables}, presented as \emph{rule diagrams}~\cite{bdg2016,bdgh2016} (i.e.\ with input graphs on the bottom, output graphs on top and with the structure of the partial injective morphisms encoded in the underlying spans of monomorphisms indicated by dashed lines):
\begin{equation*}
O_{E_{00}}:=
    \rho\left(\tP{%
        \vA{2}{1}{black}{}
        \vC{1}{1}{black}{}
        \vA{2}{2}{black}{}
        \vC{1}{2}{black}{}
        \eC{1}{1}{<}{black}{}
        \eA{2}{1}{<}{black}{}
        }\right)\,,\;
  O_{E_{01}}:=
    \rho\left(\tP{%
        \vA{2}{1}{black}{}
        \vC{1}{1}{black}{}
        \vI{1}{2}{black}{}{black}{}
        \eC{1}{1}{<}{black}{}
        \eA{2}{1}{<}{black}{}
        }\right)\,,\;
    O_{E_{10}}:=
    \rho\left(\tP{%
        \vI{1}{1}{black}{}{black}{}
        \vA{2}{2}{black}{}
        \vC{1}{2}{black}{}
        \eC{1}{1}{<}{black}{}
        \eA{2}{1}{<}{black}{}
        }\right)\,,\;
  O_{E_{11}}:=
    \rho\left(\tP{%
        \vI{1}{1}{black}{}{black}{}
        \vI{1}{2}{black}{}{black}{}
        \eI{1}{1}{<}{black}{}{black}{}
        }\right)\,.
\end{equation*}
Applying the three different observables e.g.\ to the graph state $\ket{G}=\ket{\tPgo{%
      \gE{1}{1}{1}{2}{<}{black}{}{}
      \gE{1}{2}{1}{3}{<}{black}{}{}
      \gE{1}{4}{1}{5}{<}{black}{}{}
}}$, we obtain
\begin{equation*}
O_{E_{00}}\ket{G}=1\cdot\ket{G}\,,\; O_{E_{01}}\ket{G}=O_{E_{10}}\ket{G}=2\cdot\ket{G}\,,\; O_{E_{11}}\ket{G}=3\cdot\ket{G}\,.
\end{equation*}
In other words, $O_{E_{00}}$ effectively only counts \emph{isolated edges}, $O_{E_{01}}$ and $O_{E_{01}}$ count \emph{edges whose target/source vertex is of arity $1$}, respectively, while $O_{E_{11}}$ counts \emph{edges regardless of any constraints}.
\end{exa}

For practical purposes, the concept of \emph{connected} observables plays a central role:

\begin{defi}
An observable $O_P^k\in \cO_{\bfC}$ is called \emph{connected} if and only if there do not exist observables $O_{P_1}^{k_1},O_{P_2}^{k_2}\in \cO_{\bfC}$ such that
\[
	O_P^k=O_{P_1}^{k_1}\squplus O_{P_2}^{k_2}\,.
\]
We denote the space of \emph{connected observables} by $\cO_{\bfC}^c\subsetneq \cO_{\bfC}$.
\end{defi}

\begin{restatable}{lem}{lemmaCobs}
\label{lem:Cobs}
Any observable $O_P^k\in \cO_{\bfC}$ may be equivalently expressed in terms of a \emph{polynomial in connected graph observables}, whence if $i\in\cI$ is an indexing scheme for elements of $\cO_{\bfC}^c$, there exist coefficients $f_i(O^k_P)\in \bR$ such that
\begin{equation*}
\cO^k_P=\sum_{i\in\cI}f_i(O^k_P) O_{i}\quad (O_{i}\in\cO_{\bfC}^c)
\end{equation*}
and where only finitely many of the coefficients are non-zero.
\end{restatable}
\begin{proof}
We prove the statement by induction. For generic rule algebra basis elements $R_j\equiv\delta(r_j)\in \cR_{\bfC}$ (for $j=1,2$, with $r_1,r_2\in \Lin{\bfC}$), we have that
\begin{equation*}
  R_1 \ast_{\cR_{\bfC}} R_2=R_1 \squplus  R_2+R_1\ntC_{\cR_{\bfC}} R_2\,,
\end{equation*}
where $R_1\squplus R_2$ is a shorthand for the contribution to $R_1\ast_{\cR_{\bfC}}R_2$ that arises from the composition of the underlying linear rules $r_1$ and $r_2$ along a trivial match (cf.\ Lemma~\ref{lem:rTrivMatch}), while $R_1\ntC_{\cR_{\bfC}} R_2$ denotes all terms originating from nontrivial matches. Let $\squplus$ denote the lifting of the operation $ + $ on rule algebra elements to their representations (compare Definition~\ref{def:sup}), and $\boxcoasterisk_{\cR_{\bfC}}$ the lifting of the operation $\ntC_{\cR_{\bfC}}$, whence
\begin{equation*}
\rho_{\bfC}(R_1)\boxcoasterisk_{\cR_{\bfC}}\rho_{\bfC}(R_2):=
\rho_{\bfC}(R_1)\rho_{\bfC}(R_2)-\rho_{\bfC}(R_1)\squplus\rho_{\bfC}(R_2)\,.
\end{equation*}
Then a single step of the recursion involves expressing a disconnected graph observable $O\in\cO_{\bfC}$,
\[
O\equiv \bigsquplus_{i=1}^{n_c}O_{c_i}\,,
\]
where the $O_{c_i}\in \cO_{\bfC}^c$ are connected observables, in the form
\begin{equation*}
O=
 O_{c_1} \left(\bigsquplus_{i=2}^{n_c}\cO_{c_i}\right)
-O_{c_1}\boxcoasterisk_{\cR_{\bfC}} \left(\bigsquplus_{i=2}^{n_c}\cO_{c_i}\right)\,.
\end{equation*}
As the second term has fewer connected components than $O$ (including zero if $n_c=1$), the conclusion follows by induction.
\end{proof}
Due to this Lemma, we will be able to focus our detailed analysis of stochastic dynamics on connected observables without loss of generality. The second main technical tool is the following:

\begin{thm}[Jump-closure theorem, cf.\ \cite{Behr2018}, Thm.~7.3]\label{thm:jc}
Let $\bfC$ be a finitary adhesive, extensive category. Then for arbitrary linear rules
\[
		r\equiv\rSpan{O}{o}{K}{i}{I}\in \Lin{\bfC}\,,
\]
one finds that
\begin{equation*}
	\bra{}\rho_{\bfC}(\delta(r))=\bra{}\bO(\delta(r))\,,
\end{equation*}
with the operation $\bO:\cR_{\bfC}\rightarrow \cO_{\bfC}$ from rule algebra elements to observables defined via
\begin{equation}\label{eq:bOp}
	\bO(\delta(r)):=\cO_I^i\,,\quad \cO_I^i:= \rho_{\bfC}\left(\delta\left(
		\rSpan{I}{i}{K}{i}{I}
	\right)\right)\,.
\end{equation}
We extend the operation $\mathbb{O}$ linearly in order to obtain an operation defined on generic elements of $\cR_{\bfC}$.
\begin{proof}
Due to the definition of $\bra{}$, for any $r\in \Lin{\bfC}$ and for any object $X\in \obj{\bfC}$, one finds that
\begin{equation*}
	\bra{}\rho_{\bfC}(\delta(r))\ket{X}=|\MatchSet{r}{X}|\,.
\end{equation*}
On the other hand, the admissibility of a monomorphism $m:I\hookrightarrow X$ as a match of $r$ into $X$ hinges according to Definition~\ref{def:DPOR} solely on whether or not the pushout complement of $K\xhookrightarrow{i}I\xhookrightarrow{m}X$ exists, i.e.\ is entirely independent of the precise nature of the data $O\xhookleftarrow{o}K$ of the linear rule $r$. Consequently, one obtains the same number of admissible matches if one replaces $O\xhookleftarrow{o}K$ with any other data of the form $O'\xhookleftarrow{o'}K$, and the claim follows.
\end{proof}
\end{thm}

Finally, all of these technical preparations permit us to formulate continuous-time Markov chains for stochastic rewriting systems in a uniform fashion:

\begin{thm}\label{thm:StochMech}
	Let $\bfC$ be a finitary adhesive, extensive category, and let
	\[
		\{ (\kappa_j,r_j)\}_{j\in \cJ}\,,\; \kappa_j\in \bR_{> 0}\,,\;
		r_j\in \Lin{\bfC}
	\]
	be a (finite) set of pairs of \emph{base rates} and \emph{linear rules}. Then together with an \emph{initial state} $\ket{\Psi_0}\in \Prob(\cS_{\bfC})$, this data specifies a \emph{continuous-time Markov chain} with infinitesimal generator $H\in End_{\bR}(\hat{\bfC})$ defined as
	\begin{equation}\label{eq:defH}
		H:=\hat{H}+\bar{H}\,,\quad 
		\hat{H}:= \rho_{\bfC}(h)\,,\quad
		\bar{H}:= -\bO(h)\,,\quad
		h:=\sum_{j\in \cJ}\kappa_j \delta(r_j)\,.
	\end{equation}
	We refer to $\hat{H}$ as the \emph{off-diagonal} and to $\bar{H}$ as the \emph{diagonal contribution to $H$} (for evident reasons).
	\begin{proof}
		The claim follows by verifying the defining properties of infinitesimal stochastic operators according to~\eqref{eq:Hprops}. Properties~\eqref{eq:Hprops}(i) and~\eqref{eq:Hprops}(ii) (i.e.\ the non-positivity of the diagonal and the non-negativity of the off-diagonal contributions) follow directly from the definition of $H$ as in~\eqref{eq:defH}. Property~\eqref{eq:Hprops}(iii) is equivalent to demanding that for all pure states $\ket{Y}\in \hat{\bfC}$,
		\[
  			\bra{}H\ket{Y}=0\,.
		\]
		This property is readily verified to hold as a consequence of the structure of $H$ and of the jump closure theorem (Theorem~\ref{thm:jc}), whereby $\bra{}\hat{H}=-\bra{}\bar{H}$, thus the claim follows.
	\end{proof}
\end{thm}

At this point, we would like to provide some intuitions to the interested readers via a comparison to Doi's framework for chemical reaction systems in the one-species case (see the proof of Theorem~\ref{thm:DMB} for the multi-species case):
\begin{exa}\label{ex:discrGRS}
Utilizing~\eqref{eq:raNOF} of Example~\ref{exa:HWra1}, we may provide a precise interpretation of Doi's formulation~\eqref{eq:HcrnDoi} for the infinitesimal generator $H$ of a one-species chemical reaction system:
\begin{equation*}
\begin{aligned}
H&=\sum_{i,o} \kappa_{i,o}\left(a^{\dag\:o}a^i-a^{\dag\:i}a^i\right)\\
&=\sum_{i,o} \kappa_{i,o}\left(\rho\left(\delta\rSpanAlt{\bullet^{\uplus\:o}}{}{\emptyset}{}{\bullet^{\uplus\:i}}\right)-\rho\left(\delta\rSpanAlt{\bullet^{\uplus\:i}}{}{\emptyset}{}{\bullet^{\uplus\:i}}\right)\right)
\end{aligned}
\end{equation*}
In other words, a one-species chemical reaction system is in fact verbatim a certain form of stochastic rewriting system based upon discrete graph rewriting rules (and with a state-space restricted to discrete graph states). 
\end{exa}

\subsection{Dynamics of exponential moment generating functions of observables}\label{sec:egfDynamics}

Suppose now that for a particular application at hand one were interested in studying the properties of a stochastic rewriting system via a denumerable set $\{O_i\}_{i\in\cI}$ of (w.l.o.g.\ connected) observables $O_i\in \cO_{\bfC}^c$. With $\{\lambda_i\}_{i\in \cI}$ a collection of formal variables, let us introduce the convenient multi-index notation
\begin{equation*}
	\vec{\lambda}\cdot \vec{O}:=\sum_{i\in \cI} \lambda_i O_i\,.
\end{equation*}
At least formally, in practice one would like to determine the \emph{exponential moment generating function} $M(t;\vec{\lambda})\equiv M(t;\vec{\lambda};\vec{O})$ of the system\footnote{Here and in the following, we will use the shorthand notation $M(t;\vec{\lambda})$ instead of a more verbose notation $M(t;\vec{\lambda};\vec{O})$, since the choice of the corresponding set of connected observables $\{O_i\}_{i\in \cI}$ may be assumed to be clear from the context (and with $\vec{\lambda}\cdot\vec{O}:=\sum_{i\in \cI}\lambda_i O_i$ providing the ``pairing'' of formal variables and observables).}, which is defined as
\begin{equation*}
  M(t;\vec{\lambda}):=\left\langle e^{\vec{\lambda}\cdot \vec{O}}\right\rangle(t)\equiv \left\langle\left\vert e^{\vec{\lambda}\cdot \vec{O}}\right\vert\Psi(t)\right\rangle\,.
\end{equation*}
Here, $\ket{\Psi(t)}=\cE(t)\ket{\Psi(0)}$ is the time-dependent state of the stochastic rewriting system, related to the Master equation via $\frac{d}{dt}\ket{\Psi(t)}=H\ket{\Psi(t)}$. Clearly, a closed-form solution for $M(t;\vec{\lambda})$ would amount to the knowledge of all moments of the observables of interest at all times $t\geq 0$ along the stochastic evolution, according to
\begin{equation*}
  \langle O_{i_1}\dotsc O_{i_n}\rangle(t)=\left[\tfrac{\partial}{\partial \lambda_{i_1}}\dotsc \tfrac{\partial}{\partial \lambda_{i_n}}M(t;\vec{\lambda})\right]\bigg\vert_{\vec{\lambda}\to \vec{0}}\,.
\end{equation*}
Disregarding for the moment some potential problems with the formal exponential moment generating function (such as the question of finiteness of all the moments, summability of terms etc., which in practice have to be dealt with on a case-by-case basis; see however~\cite{DBLP:conf/fossacs/DanosHGS17}), one may derive from the generic properties of a stochastic rewriting system and the aforementioned results the following formal\footnote{Note in particular that we do not make any claims here on the existence of solutions of these formal differential equations, which is why such equations at present must be analyzed on a case-by-case basis. We will present a number of concrete examples where the analysis is possible in Sections~\ref{sec:AE} and~\ref{sec:BC}.} evolution equations:
\begin{thm}[Compare~\cite{bdg2016}]\label{thm:EMGFev}
Let $H$ denote the evolution operator of a stochastic rewriting system, and let $O_1$, $O_2$, $\dotsc$ ($O_i\in \cO_{\bfC}^c$) be an at most denumerable set of connected observables. Then the exponential moment generating function $M(t;\vec{\lambda})$ fulfills the following \emph{formal evolution equation}:
\begin{equation*}
  \label{eq:MGev}
  \begin{aligned}
    \frac{\partial}{\partial t}M(t;\vec{\lambda})&=\left\langle \left(e^{ad_{\vec{\lambda}\cdot \vec{O}}}H\right)e^{\vec{\lambda}\cdot \vec{O}}\right\rangle(t)
    =\sum_{n\geq1}\frac{1}{n!}\left\langle \left(ad_{\vec{\lambda}\cdot\vec{O}}^{\circ\:n}\hat{H}\right)e^{\vec{\lambda}\cdot\vec{O}}\right\rangle(t)\,.
  \end{aligned}
\end{equation*}
Here, the so-called \emph{adjoint action} of a linear operator $A$ on a linear operator $B$ (with $A$ and $B$ composable) is defined as $ad_A B:=AB-BA$, also referred to as the \emph{commutator} $[A,B]=AB-BA$ of $A$ and $B$. By convention, $ad_A^{\circ 0}B:=B$, while $ad_A^{\circ\:n}B$ for $n>0$ denotes the $n$-fold application of $ad_A$ to $B$.
\begin{proof}
By making use of the Master equation and (in the step marked $(*)$) of a variant of the BCH formula (see e.g.\  \cite{hall2015lieGroups}, Prop.~3.35), we may compute:
  \begin{align*}
    \frac{d}{dt}\left\langle e^{\vec{\lambda}\cdot \vec{O}}\right\rangle(t)
    &=\left\langle e^{\vec{\lambda}\cdot \vec{O}}H\right\rangle(t)=\left\langle \left(e^{\vec{\lambda}\cdot \vec{O}}He^{\vec{-\lambda}\cdot \vec{O}}\right)e^{\vec{\lambda}\cdot \vec{O}}\right\rangle(t)\overset{(*)}{=}\left\langle \left(e^{ad_{\vec{\lambda}\cdot \vec{O}}}H\right)e^{\vec{\lambda}\cdot \vec{O}}\right\rangle(t)\,.
  \end{align*}%
   Expanding the exponential and using that $ad^{\circ\:0}_A B:=B$ as well as $\bra{}H=0$ (which follows from Lemma~\ref{lem:Haux}, and also explicitly from Theorem~\ref{thm:jc}) results in the expanded form of the evolution equation.
\end{proof}
\end{thm}

Consequently, the evolution of observables is entirely governed by \emph{static} combinatorial-algebraic relationships as encoded in the nested commutators in~\eqref{eq:MGev} of Theorem~\ref{thm:EMGFev}. We invite the interested readers to consult Example~\ref{exa:bd2} for a concrete case of the moment EGF evolution equation in the case of a one-species chemical reaction system.

\section{Combinatorial conversion for stochastic rewriting systems}
\label{sec:CC}

Combining the result on the evolution of the exponential moment generating function $M(t;\vec{\lambda})$ for a chosen subset $\mathsf{O}=\{O_1,O_2,\dotsc\}$ of observables as presented in Theorem~\ref{thm:EMGFev} with the result on jump-closure of observables (Theorem~\ref{thm:jc}), we encounter the following fundamental problem: each contribution to the evolution equation~\eqref{eq:MGev} of the form
\begin{equation*}
\left\langle \left\vert\left(ad_{\vec{\lambda}\cdot\vec{O}}^{\circ\:n}\hat{H}\right)\right.\right.
\end{equation*}
may be transformed by virtue of the jump-closure theorem into a linear combination of terms of the form 
$\bra{}O'_k$ for $O'_k\in \cO_{\bfC}$. In general, these terms may involve new observables not in 
our chosen subset $\mathsf{O}$.

We will now focus on a class of special cases in which this complication is absent. This will allow for making contact between traditional techniques of analyzing exponential moment generating functions known from the statistics literature and SRSs.

\begin{defi}\label{def:PJC}
Consider a stochastic DPO rewriting system over a finitary adhesive, extensive category $\bfC$ with evolution operator $H$. Then we refer to a set $\mathsf{O}$ of connected graph observables,
\[
  \mathsf{O}=\{O_i\in \cO_{\bfC}^c\}_{i\in\cI}\,,
\]
for some (possibly countably infinite) index set $\cI$ as \textbf{\emph{polynomially jump-closed}} if and only if all higher moments of the observables $O_i$ have evolution equations as in~\eqref{eq:MGev} that are expressible in terms of polynomials in the observables $O_i$. More concretely, this entails the \emph{polynomial jump closure $(\mathsf{PJC})$ condition} namely that for all $n$
\begin{equation*}
(\mathsf{PJC})\quad    
\exists \vec{N(n)}\in \bZ_{\geq 0}^{\cI}, \gamma_n(\vec{\lambda};\vec{k})\in \bR:\;\bra{}ad_{\vec{\lambda}\cdot\vec{O}}^{\circ\: n}H=
  \sum_{\vec{k}=\vec{0}}^{\vec{N(n)}}\gamma_n(\vec{\lambda};\vec{k})\bra{}\vec{O}^{\vec{k}}\,.
\end{equation*}
Here, the multi-index notation $\vec{O}^{\vec{k}}$ denotes a product of linear operators $O_i$ according to
\[
\vec{O}^{\vec{k}}:=\prod_{i\in\cI}O_i^{k_i}\,.
\]
Since the definition also covers the trivial case where $\cI$ indexes \emph{all} basis elements of the space $\cO_{\bfC}^c$ of connected graph observables, we sometimes emphasize the non-trivial cases where $\cI$ is either a finite set, or where $\{O_i\}_{i\in\cI}$ is a strict subset of the set of basis elements of $\cO_{\bfC}^c$.
\end{defi}
The readers interested in a heuristic motivation for the concept of polynomial jump-closure is invited to consult the proof of Theorem~\ref{thm:DMB}, in which this property is illustrated for the special case of (multi-species) chemical reaction systems that coincide with a certain type of stochastic rewriting systems over discrete graphs.\\

The property of polynomial jump-closure has the following remarkable consequence for a stochastic rewriting system, which constitutes one of the main results of this paper:

\begin{restatable}[The Combinatorial Conversion Theorem]{thm}{thmCCT}
\label{thm:CCT}
For a polynomially jump-closed set of observables $\mathsf{O}=\{O_i\}_{i\in\cI}$ (with $O_i\in\cO_{\bfC}^c$ for all $i\in\cI$) of a stochastic rewriting system over a finitary adhesive, extensive category $\bfC$ with evolution operator $H$, the evolution equation for the exponential moment generating function $M(t;\vec{\lambda})$ of $\mathsf{O}$ may be converted from its explicit expression in the observables $O_i$ into a \emph{partial differential equation} of $M(t;\vec{\lambda})$ itself w.r.t.\ the formal parameters $\{\lambda_i\}_{i\in\cI}$:
\begin{equation}\label{eq:CCTmain}
\begin{aligned}
\tfrac{\partial}{\partial t}M(t;\vec{\lambda})&=\bD(\vec{\lambda},\partial_{\vec{\lambda}})M(t;\vec{\lambda})\,,\qquad
\bD(\vec{\lambda},\partial_{\vec{\lambda}})=
\left(\left[ \bra{}\left(e^{ad_{\vec{\lambda}\cdot\vec{O}}}H\right)\right]\bigg\vert_{O_i\to\tfrac{\partial}{\partial\lambda_i}}\right)\ket{\emptyset}\,.
\end{aligned}
\end{equation}
Here, in the definition of the differential operator $\bD$, we have made use of the assumption of polynomial jump-closure in converting the expression in square brackets into $\bra{}$ applied to a formal series in the $O_i$.
\end{restatable}
\begin{proof}
The proof follows immediately from the following application of the polynomial jump-closure assumption: with $\vec{\lambda}\cdot\vec{O}=\sum_{i\in\cI}\lambda_i O_i$, polynomial jump closure of $\mathsf{O}$ w.r.t.\ $H$ implies that
\begin{equation*}
\label{eq:ProofThm42A}
\begin{aligned}
  \forall n\in \bZ_{>0}:\;\bra{}ad_{\vec{\lambda}\cdot\vec{O}}^{\circ\: n}He^{\vec{\lambda}\cdot\vec{O}}&=
  \sum_{\vec{k}=\vec{0}}^{\vec{N(n)}}\gamma_n(\vec{\lambda};\vec{k})\bra{}\vec{O}^{\vec{k}}e^{\vec{\lambda}\cdot\vec{O}}\\
  &=
  \sum_{\vec{k}=\vec{0}}^{\vec{N(n)}}\gamma_n(\vec{\lambda};\vec{k})\left(\vec{\tfrac{\partial}{\partial \lambda}}\right)^{\vec{k}}\bra{}e^{\vec{\lambda}\cdot\vec{O}}\,.
\end{aligned}
\end{equation*}
Here, we made use of the convenient multi-index notation
\begin{equation*}
\left(\vec{\tfrac{\partial}{\partial \lambda}}\right)^{\vec{k}}:=\prod_{i\in\cI}\left(\tfrac{\partial}{\partial \lambda_i}\right)^{k_i}\,,
\end{equation*}
such that in particular
\[
\left(\vec{\tfrac{\partial}{\partial \lambda}}\right)^{\vec{k}}e^{\vec{\lambda}\cdot\vec{O}}=\vec{O}^{\vec{k}}e^{\vec{\lambda}\cdot\vec{O}}\,.
\]
Invoking Theorem~\ref{thm:EMGFev}, we find via~\eqref{eq:ProofThm42A} an evolution equation of the form
\begin{equation*}
  \begin{aligned}
    \frac{\partial}{\partial t}M(t;\vec{\lambda})&=
    \left.\left.\left\langle\left\vert \left(e^{ad_{\vec{\lambda}\cdot \vec{O}}}H\right)e^{\vec{\lambda}\cdot \vec{O}}\right\vert \Psi(t)\right\rangle\right.\right.
    =\sum_{n\geq1}\frac{1}{n!}
    \left.\left.\left\langle\left\vert \left(ad_{\vec{\lambda}\cdot\vec{O}}^{\circ\:n}\hat{H}\right)e^{\vec{\lambda}\cdot\vec{O}}\right\vert\Psi(t)\right\rangle\right.\right.\\
    &=\sum_{n\geq 1}\frac{1}{n!}\sum_{\vec{k}=\vec{0}}^{\vec{N(n)}}\gamma_n(\vec{\lambda};\vec{k})\left(\vec{\tfrac{\partial}{\partial \lambda}}\right)^{\vec{k}}
    \left.\left.\left\langle\left\vert e^{\vec{\lambda}\cdot\vec{O}}\right\vert\Psi(t)\right\rangle\right.\right.
    \,.
  \end{aligned}
\end{equation*}
Finally, the auxiliary identity
\[
    \bra{}e^{\vec{\lambda}\cdot\vec{O}}\ket{\emptyset}=\braket{}{\emptyset}=1
\]
permits to verify that
\begin{equation*}
    \bD(\vec{\lambda},\partial_{\vec{\lambda}}):=\sum_{n\geq 1}\frac{1}{n!}\sum_{\vec{k}=\vec{0}}^{\vec{N(n)}}\gamma_n(\vec{\lambda};\vec{k})\left(\vec{\tfrac{\partial}{\partial \lambda}}\right)^{\vec{k}}
    =\left(\left[ \bra{}\left(e^{ad_{\vec{\lambda}\cdot\vec{O}}}H\right)\right]\bigg\vert_{O_i\to\tfrac{\partial}{\partial\lambda_i}}\right)\ket{\emptyset}
    \,,
\end{equation*}
which concludes the proof.
\end{proof}

While we will present some concrete examples of combinatorial conversion for non-discrete graph rewriting systems in Sections~\ref{sec:AE} and~\ref{sec:BC}, it is worthwhile emphasizing that the concept itself is modeled after the type of computation relevant in the determination of moment EGF evolution equations in chemical reaction systems (see~\eqref{eq:MEGFbd} of Example~\ref{exa:bd2}).

\section{Moment bisimulation}
\label{sec:DB}

The manner in which  the Combinatorial Conversion Theorem describes the evolution of connected observables by the evolution of their exponential moment generating function, and whence of a formal power series, motivates the following definition:

\begin{defi}[Moment bisimulation]
Consider two SRSs with Hamiltonians $H_i$ and two equinumerous sets $\mathsf{O}_i$ of connected observables polynomially jump-closed w.r.t.\ $H_i$, respectively ($i\in 1,2$). Denote by $f:\mathsf{O}_1\xrightarrow{\cong}\mathsf{O}_2$ a bijection of the two sets of observables. Then the pairs $(H_1,\mathsf{O}_1)$ and $(H_2,\mathsf{O}_2)$ are said to be \textbf{\emph{moment bisimilar (via $f$)}} if the \emph{moment bisimilarity $(\mathsf{MB})$ condition} holds:
\begin{equation*}
\left(\left[ \bra{}\left(e^{ad_{\vec{\lambda}\cdot\vec{O}}}H_1\right)\right]\bigg\vert_{O_i\to\tfrac{\partial}{\partial\lambda_i}}\right)\ket{\emptyset}=\left(\left[ \bra{}\left(e^{ad_{\vec{\lambda}\cdot\vec{f(O)}}}H_2\right)\right]\bigg\vert_{f(O_i)\to\tfrac{\partial}{\partial\lambda_i}}\right)\ket{\emptyset}\,.
\tag*{\textsf{(MB)}}
\end{equation*}
Here, we have employed the multi-index notations
\[
\vec{\lambda}\cdot \vec{O}:=\lambda_1 O_1+\dotsc +\lambda_n O_n\,,\qquad \vec{\lambda}\cdot \vec{f(O)}:=\lambda_1  f(O_1)+\dotsc +\lambda_nf(O_n)\,,
\]
with $n=|\mathsf{O}_1|=|\mathsf{O}_2|$, and where for each $j\in\{1,\dotsc,n\}$, $f(O_j)\in\mathsf{O}_2$ denotes the image of $O_j\in \mathsf{O}_1$ under the isomorphism $f:\mathsf{O}_1\xrightarrow{\cong}\mathsf{O}_2$.
\end{defi}

By virtue of the Combinatorial Conversion Theorem (Theorem~\ref{thm:CCT}), in the situation described in the definition the moment bisimilarity condition $(\mathsf{MB})$, letting \[
M_1(t;\vec{\lambda}):=\left.\left.\left\langle\left\vert e^{\vec{\lambda}\cdot\vec{O}}\right\vert \Psi_1(t)\right\rangle\right.\right.\quad \text{and}\quad M_2(t;\vec{\lambda}):=\left.\left.\left\langle\left\vert e^{\vec{\lambda}\cdot\vec{f(O)}}\right\vert \Psi_2(t)\right\rangle\right.\right.
\]
entails that
\begin{equation*}
\tfrac{\partial}{\partial t}M_1(t;\vec{\lambda})=\tfrac{\partial}{\partial t}M_2(t;\vec{\lambda})\,.
\end{equation*}
Therefore, for choices of initial states $\ket{\Psi_1(0)},\ket{\Psi_2(0)}\in \Prob(\cS)$ such that $M_1(0;\vec{\lambda})=M_2(0;\vec{\lambda})$ one finds that $M_1(t;\vec{\lambda})=M_2(t;\vec{\lambda})$ for all $t\geq 0$.\\

As a first step towards a constructive theory of moment bisimulation, we present in the following a peculiar scenario in which a polynomially jump-closed set of connected observables $\mathsf{O}$ for a stochastic rewriting system (SRS) with Hamiltonian $H$ is moment bisimilar to a set $\mathsf{O}_{discr}$ of \emph{discrete} connected observables for a Hamiltonian $H_{discr}$ of a \emph{discrete} SRS. The idea originated from previous work on the study of exponential moment generating function evolution equations for chemical reaction systems~\cite{bdp2017}, which uncovered the precise form such evolution equations may take (and thus a characterization of the \emph{discrete moment bisimilarity class}). The interested readers may wish to compare the following statements to the discussion of discrete graph rewriting systems (i.e.\ chemical reaction systems) as provided in Section~\ref{sec:CRS}.
\begin{restatable}[Discrete Moment Bisimulation]{thm}{dBSthm}\label{thm:DMB}
Let $H=\sum_{j\in \cJ}\kappa_j(\rho(h_j)-\bO(h_j))$ be the Hamiltonian of a SRS, and let $\mathsf{O}=\{O_i\in \cO_c\}_{i\in \cI}$ be a polynomially jump-closed set of observables for $H$. Suppose the following two conditions\footnote{One may further relax condition $(i)$ to allow for constant coefficients $\eta_{ij}=\beta_i\cdot \eta'_{ij}$ with $\beta_i\in \bR_{\neq 0}$ and $\eta_{ij}'\in \bZ$, at the expense of a change of formal variables $\lambda_i$.} (amounting to the  \emph{discrete moment bisimulation $(\mathsf{DMB})$ condition}) are verified:
\begin{equation*}
\begin{aligned}
(i) && \forall i\in \cI,j\in \cJ:\exists \eta_{ij}\in \bZ:\quad   ad_{O_i}(\rho(h_j))&=\eta_{ij}\rho(h_j)\\
(ii)&& \forall j\in \cJ:\exists \alpha_j\in \bR,\vec{k}_j\in \bZ_{\geq 0}^{\cI}:
\quad \bra{}\rho(h_j)&=\alpha_j\sum_{\vec{n}=\vec{0}}^{\vec{k}_j}
s_1(\vec{k}_j,\vec{n})\bra{}\vec{O}^{\vec{n}}\,,
\end{aligned}
\tag*{\textsf{(DMB)}}
\end{equation*}
where we made use of the multi-index notation\footnote{Here, $s_1(m,n)$ denotes the (signed) Stirling numbers of the first kind, cf.\ sequence~\href{https://oeis.org/A008275}{A008275} of the OEIS database~\cite{OEISstirlingFirst}); compare~\eqref{eq:discrStirling}.}
\begin{equation*}
    s_1(\vec{k}_j,\vec{n}):=\prod_{i\in\cI} s_1((\vec{k}_j)_i;n_i)\,.
\end{equation*}
Then for every bijection $F:\cI\xrightarrow{\cong}\cC$ from $\cI$ to a set of \emph{vertex colors} $\cC$, denoting by $\mathsf{O}_{discr}:=\{\hat{n}_{c}\}_{c\in \cC}$ a set of discrete connected graph observables (with $\hat{n}_c\in \cO_{discr}$ counting vertices of color $c\in \cC$) and by $H_{discr}$ the Hamiltonian of an SGRS of discrete graphs with vertices of colors $\cC$ defined as
\begin{equation*}
\begin{aligned}
H_{discr}&:=\sum_{j\in \cJ}\alpha_j\kappa_j(\rho(\tilde{h_j})-\bO(\tilde{h_j}))\,,\; \tilde{h}_j:=
\delta\rSpanAlt{\vec{\bullet}^{(\vec{\eta}_j+\vec{k}_j)}}{}{\emptyset}{}{\vec{\bullet}^{\vec{k}_j}}\,,\;\vec{\eta}_{j_i}:=\eta_{ij}\,.
\end{aligned}
\end{equation*}
the pair $(H,\mathsf{O})$ is moment bisimilar to $(H_{discr},\mathsf{O}_{discr})$ via the isomorphism $f(O_i):=\hat{n}_{F(i)}$.
\end{restatable}
\begin{proof}
The proof relies on results from the theory of stochastic mechanics for chemical reaction systems recently developed by the first-named author and his team~\cite{bdp2017}. Let thus $\cC$ be a set of colors, and let
\[
  a^{\dag}_c:=
  \rho\left(\delta\rSpanAlt{\bullet_c}{}{\emptyset}{}{\emptyset}\right)
  \quad\text{and}\quad a_c:=\rho\left(\delta\rSpanAlt{\emptyset}{}{\emptyset}{}{\bullet_c}\right)
\]
denote the \emph{``creation''} and \emph{``annihilation''} operators, respectively (of vertices of color $c\in \cC$; see also~\cite{bdg2016,Behr2018} for an extended discussion). According to these definitions, the operators fulfill the \emph{canonical commutation relations} of the multi-species Heisenberg-Weyl algebra, whence for all $c,c'\in \cC$,
\[
[a^{\dag}_c,a^{\dag}_{c'}]=0=[a_c,a_{c'}]\,,\quad
[a_c,a^{\dag}_{c'}]=\delta_{c,c'}\mathbb{1}\,.
\]
Let $\ket{\vec{\bullet}^{ + \:\vec{n}}}$ denote the basis vector associated to the finite discrete graph with $|\vec{n}|:=\sum_{c\in \cC}n_c$ vertices (with $n_c$ vertices of color $c\in \cC$), and $\hat{n}_c:=a^{\dag}_ca_c$ the so-called \emph{number operator of color $c$}, which by definition is effectively a graph observable ``counting'' vertices of color $c\in \cC$,
\[
\hat{n}_c\ket{\vec{\bullet}^{ + \:\vec{n}}}=n_c\ket{\vec{\bullet}^{ + \:\vec{n}}}\,.
\]
It may be checked from our definitions of rule composition that one may w.l.o.g.\ restrict the set of linear rules of discrete graphs acting on discrete graphs only to the standard set of linear rules giving rise to rule algebra elements of the form
\[
h_{\vec{o},\vec{i}}:=
\delta\rSpanAlt{\vec{\bullet}^{ + \:\vec{o}}}{}{\emptyset}{}{\vec{\bullet}^{ + \:\vec{i}}}\,,
\]
since for all $\vec{i},\vec{o},\vec{n}\in \bZ_{\geq 0}^{\cC}$ and for all injective partial morphisms 
\[
\left(\GRule{\vec{\bullet}^{ + \:\vec{o}}}{t}{\vec{\bullet}^{ + \:\vec{i}}}\right)=\rSpanAlt{\vec{\bullet}^{ + \:\vec{o}}}{}{\vec{\bullet}^{ + \:\vec{k}}}{}{\vec{\bullet}^{ + \:\vec{i}}}\,,
\]
one finds that
\[
\rho\left(\delta\rSpanAlt{\vec{\bullet}^{ + \:\vec{o}}}{}{\vec{\bullet}^{ + \:\vec{k}}}{}{\vec{\bullet}^{ + \:\vec{i}}}\right)\ket{\vec{\bullet}^{ + \:\vec{n}}}=
\rho\left(\delta\rSpanAlt{\vec{\bullet}^{ + \:\vec{o}}}{}{\emptyset}{}{\vec{\bullet}^{ + \:\vec{i}}}\right)\ket{\vec{\bullet}^{ + \:\vec{n}}}\,.
\]
Using the canonical commutation relations to verify the formula
\[
H_{\vec{o},\vec{i}}:=\rho(h_{\vec{o},\vec{i}})=\vec{a}^{\dag\:\vec{o}}\vec{a}^{\vec{i}}\,,
\]
we thus arrive at the first conclusion that we have identified the in this sense canonical basis for infinitesimal generators of discrete graph rewriting CTMCs. It remains to derive the precise form of the algebraic conditions as presented in the theorem. To this end, applying yet again the canonical commutation relations and making use of the jump-closure theorem, we derive the following two properties of discrete rewriting systems:
\begin{align*}
(i)&&\quad 
[\hat{n}_c,H_{\vec{o},\vec{i}}]&=(o_c-i_c)H_{\vec{o},\vec{i}}\\
(ii)&& \quad 
\bra{}H_{\vec{o},\vec{i}}
&=\bra{}\vec{a}^{\dag\:\vec{i}}\vec{a}^{\vec{i}}=\sum_{\vec{k}=\vec{0}}^{\vec{i}}\left(\prod_{c\in \cC}s_1(i_c,k_c)\right)\bra{}\vec{\hat{n}}^{\vec{k}}
\end{align*}%
Here, $(ii)$ may be seen as the multi-colored variant of~\eqref{eq:discrStirling}. The claim of the theorem then follows by requiring the candidate SGRS and its polynomially jump-closed set of observables to satisfy the equivalent of equations $(i)$ and $(ii)$ above for a choice of ``index-to-color'' bijection $f:\cI\xrightarrow{\cong}\cC$.
\end{proof}

The key point of this theorem is that one may focus on computing the evolution of $\cM_{discr}(t;\vec{\lambda})$ instead of on the more difficult task of computing $\cM(t;\vec{\lambda})$ directly. In particular, this opens the possibility of using e.g.\ standard SSA techniques to obtain (approximations of) $\cM_{discr}(t;\vec{\lambda})$ or individual moments by means of \emph{simulations}, or to employ in favorable cases exact solution techniques such as the ones developed in~\cite{bdp2017} for solving the evolution equation for $\cM_{discr}(t;\vec{\lambda})$ (and whence, if $\cM_{discr}(0;\vec{\lambda})=\cM(0;\vec{\lambda})$, for $\cM(t;\vec{\lambda})$). Intriguingly, this opens the possibility to obtain \emph{probability distributions} of observable counts, a feature hardly ever realizable in practice via direct simulation of stochastic graph rewriting systems. While as explained in the proof of this theorem any \emph{discrete} graph observable may be expressed as a polynomial in the discrete graph observables $\hat{n}_c$ with precisely the Stirling number coefficients as demanded by the theorem, this feature does not hold in general for generic graph observables, which renders the search for realizations of discrete moment bisimulations a highly non-trivial task.

\section{Application example case study: a Voter Model}
\label{sec:AE}

Consider the following variant of a \emph{Voter Model}~\cite{danos2014approximations}, constructed as a stochastic graph rewriting system over the (colored version of the) finitary adhesive, extensive category $\mathbf{uGraph}$~\cite{Behr2018} of finite undirected multigraphs with the following four transitions:
\gdef\tpScale{0.6}
\begin{equation*}
\GRule{\tP{%
\node[vertices,fill=black!40!white] (u) at (0,1) {};
\node[vertices,fill=white] (a) at (\colSEP,1) {};
\node[vertices] (b) at (2*\colSEP,1) {};
\draw (a) edge[undirEdge,black] (b);}}{\:\kappa_{0}\:}{\tP{%
\node[vertices,fill=black!40!white] (u) at (0,1) {};
\node[vertices,fill=white] (a) at (\colSEP,1) {};
\node[vertices] (b) at (2*\colSEP,1) {};
\draw (a) edge[undirEdge,black] (u);
}}\,,\quad
\GRule{\tP{%
\node[vertices,fill=black!40!white] (u) at (0,1) {};
\node[vertices] (a) at (\colSEP,1) {};
\node[vertices,fill=white] (b) at (2*\colSEP,1) {};
\draw (a) edge[undirEdge,black] (b);}}{\:\kappa_{1}\:}{\tP{%
\node[vertices,fill=black!40!white] (u) at (0,1) {};
\node[vertices] (a) at (\colSEP,1) {};
\node[vertices,fill=white] (b) at (2*\colSEP,1) {};
\draw (a) edge[undirEdge,black] (u);
}}\,,\quad
\GRule{\tP{%
\node[vertices,fill=white] (a) at (1,1) {};
\node[vertices] (b) at (1+\colSEP,1) {};
\draw (a) edge[undirEdge,black] (b);}}{\:\kappa_{01}\:}{\tP{%
\node[vertices,fill=white] (a) at (1,1) {};
\node[vertices,fill=white] (b) at (1+\colSEP,1) {};
\draw (a) edge[undirEdge,black] (b);}}\,,\quad
\GRule{\tP{%
\node[vertices,fill=white] (a) at (1,1) {};
\node[vertices] (b) at (1+\colSEP,1) {};
\draw (a) edge[undirEdge,black] (b);}}{\:\kappa_{10}\:}{\tP{%
\node[vertices] (a) at (1,1) {};
\node[vertices] (b) at (1+\colSEP,1) {};
\draw (a) edge[undirEdge,black] (b);}}\,.
\end{equation*}
In this shorthand notation, vertices marked $\tP{%
\node[vertices,fill=black!40!white] (u) at (0,1) {};}$ are of either black or white color. Translating the graphical rules into rule algebra elements via the standard dictionary presented in this paper, we obtain the following explicit realization of the model. 

\paragraph{\textbf{Notational convention:}} In visualizations of rule algebra elements associated to linear rules via rule diagrams, we take the convention that ``time flows upwards'', i.e.\ the bottom graph denotes the input $I$ and the top graph the output $O$ of the respective linear rule $(O\leftarrow K\rightarrow I)$ depicted a given diagram, while the dotted lines indicate which elements are preserved throughout the transformation (i.e.\ the structure of the graph $K$ and of the embeddings $K\rightarrow I$ and $K\rightarrow O$). For visual clarity, elements of $O$ and $I$ not in the codomains of $K\rightarrow O$ and $K\rightarrow I$, respectively, will be marked by in- and outgoing dotted lines decorated with a $\times$ symbol. As a notational simplification, we will moreover from hereon no longer indicate graphically whether a given edge is preserved vs.\@ deleted and then recreated, since the two corresponding rules are \emph{indistinguishable} in their action on graphs according to our notion of canonical representation. Consequently, we may lighten our graphical notations by assuming that edges on the input of a rule are always deleted, while edges on the output are always created, which permits us to drop the dotted lines that were previously used to indicate these fine details for the edges altogether.
\begin{defi}
The \emph{Voter Model} is a stochastic DPO graph rewriting system with evolution operator (with the operation $\bO$ as defined in~\eqref{eq:bOp})
\begin{equation*}
\label{eq:defVM}
\begin{aligned}
  H_{VM}&:=\rho(h_{VM})-\bO(h_{VM})\\
  h_{VM}&:=\kappa_0 (h_{0_w}+h_{0_b})+\kappa_1(h_{1_w}+h_{1_b})+\kappa_{01}h_{01}+\kappa_{10}h_{10}\\
  h_{0_w}&:=\includeFig{1}{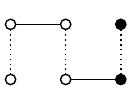}\,,\quad 
  h_{0_b}:=\includeFig{2}{images/defVM.pdf}\,,\quad
  h_{1_w}:=\includeFig{3}{images/defVM.pdf}\,,\quad
  h_{1_b}:=\includeFig{4}{images/defVM.pdf}\\
  h_{01}&:=\includeFig{5}{images/defVM.pdf}\,,\quad
  h_{10}:=\includeFig{6}{images/defVM.pdf}\\
  \bO(h_{0_w})&=\bO(h_{1_w})=O_{\tPgo{%
  \node[vertices,fill=white] (a) at (1,1) {};
  }\;\tPgo{%
\node[vertices,fill=white] (a) at (0,1) {};
\node[vertices] (b) at (\colSEP,1) {};
\draw (a) edge[undirEdge,black] (b);}}\,,\quad
  \bO(h_{0_b})=\bO(h_{1_b})=O_{\tPgo{%
  \node[vertices] (a) at (1,1) {};
  }\;\;\tPgo{%
\node[vertices,fill=white] (a) at (0,1) {};
\node[vertices] (b) at (\colSEP,1) {};
\draw (a) edge[undirEdge,black] (b);}}\\
\bO(h_{01})&=\bO(h_{10})=O_{\tPgo{%
\node[vertices,fill=white] (a) at (0,1) {};
\node[vertices] (b) at (\colSEP,1) {};
\draw (a) edge[undirEdge,black] (b);}}\,.
\end{aligned}
\end{equation*}
Here, we have made use of the shorthand notation $O_P:=\rho(\delta(P\hookleftarrow P'\hookrightarrow P))$, with $P\in \obj{\bfC}_{\cong}$ a ``pattern'', and for $P'$ obtained from $P$ by dropping all edges (which is the shape of observables obtained by applying $\bO$ to the contributions of $H$ in the present case).
\end{defi}

Our aim in studying this model is to illustrate the utility of the technique of exponential moment generating function approach in interplay with the rule algebra formalism. To this end, let us study the time-evolution of the exponential moment generating function of the following graph observables (counting vertices and edges of the different colors):
\begin{align*}
  O_w&:=O_{\tPgo{%
  \node[vertices,fill=white] (a) at (1,1) {};
  }}\,,\quad
  O_b:=O_{\tPgo{%
  \node[vertices] (a) at (1,1) {};
  }}\,,\quad
  O_{ww}:=\tfrac{1}{2}O_{\tPgo{%
\node[vertices,fill=white] (a) at (0,1) {};
\node[vertices,fill=white] (b) at (\colSEP,1) {};
\draw (a) edge[undirEdge,black] (b);}}\,,\quad
  O_{wb}:=O_{\tPgo{%
\node[vertices] (a) at (0,1) {};
\node[vertices,fill=white] (b) at (\colSEP,1) {};
\draw (a) edge[undirEdge,black] (b);}}\,,\quad
  O_{bb}:=\tfrac{1}{2}O_{\tPgo{%
\node[vertices] (a) at (0,1) {};
\node[vertices] (b) at (\colSEP,1) {};
\draw (a) edge[undirEdge,black] (b);}}\,.
\end{align*}%
Making use of rule-algebraic compositions, we may first of all convert the expressions for $\bO(h_{i_w})$ and $\bO(h_{i_b})$ ($i=0,1$) into the form
\begin{equation}\label{eq:VMjc}
\bO(h_{i_w})=(O_w-1)O_{wb}\,,\quad \bO(h_{i_b})=(O_b-1)O_{wb}\,.
\end{equation}
In the analysis presented in the following, we will feature a number of typical characteristics of such computations, which here include the appearances of certain \emph{symmetries} due to conservation laws (of total numbers of vertices and of edges, respectively).\\

To provide a simple, yet non-trivial example of discrete moment bisimulation, let us consider the case $\kappa_{01}=\kappa_{10}=0$, i.e.\ the Voter Model without vertex-recoloring transitions. Denote by $H_F:=H_{VM}\vert_{\kappa_{01}=\kappa_{10}=0}$ the corresponding Hamiltonian. It is evident that the remaining (``edge-flipping'') transitions, do not modify the numbers of vertices of white and black color, respectively, a feature that manifests itself upon computing the exponential moment generating function for the vertex observables $O_w$ and $O_b$:
\begin{equation}\label{eq:vertexCons}
\begin{aligned}
ad_{\lambda_wO_w+\lambda_b O_b}(H_F)&=0\quad\xRightarrow{\eqref{eq:MGev}}\quad
\tfrac{\partial}{\partial t}\cM_{F_v}(t;\lambda_w,\lambda_b):=\tfrac{\partial}{\partial t}\bra{}e^{\lambda_wO_w+\lambda_b O_b}\ket{\Psi(t)}=0\\
\Leftrightarrow\quad \cM_{F_v}(t;\lambda_w,\lambda_b)&=\cM_{F_v}(0;\lambda_w,\lambda_b)
\,.
\end{aligned}
\end{equation}
If we initialize the system at a \emph{pure} graph state, i.e.\ for $\ket{\Psi(0)}=\ket{G_0}$ (for some graph $G_0$), for which $\bra{}O_x\ket{G_0}=N_{x}$ ($x\in \{w,b\}$) are the numbers of white and black vertices, respectively, we find the explicit formula
\begin{equation*}
\cM_{F_v}(t;\lambda_w,\lambda_b)=\cM_{F_v}(0;\lambda_w,\lambda_b)=e^{\lambda_w N_w+\lambda_b N_b}\,,
\end{equation*}
which in turn implies that $O_x\ket{\Psi(t)}=N_x\ket{\Psi(t)}$ for $x\in \{w,b\}$ and for all $t\geq 0$.\\ 

To determine the dynamics of edge observables, we introduce the notations
\begin{equation*}
\begin{aligned}
\vec{\lambda_E}\cdot\vec{O_E}&:=\lambda_{ww}O_{ww}+\lambda_{wb}O_{wb}+\lambda_{bb}O_{bb}\\
\vec{\kappa_F}\cdot \vec{H_F}&:=
\kappa_0(H_{0_w}+H_{0_b})+\kappa_1(H_{1_w}+H_{1_b})
\end{aligned}
\end{equation*}
with $H_X:=\rho(h_X)$ and compute the following commutator (using the rule algebra):
\begin{equation}\label{eq:edgeComm}
[\vec{\lambda_E}\cdot\vec{O_E},\vec{\kappa_F}\cdot \vec{H_F}]=\kappa_0(\lambda_{ww}-\lambda_{wb})H_{0_w}+\kappa_1(\lambda_{bb}-\lambda_{wb})H_{1_b}\,.
\end{equation}
We observe in particular that for $O_E:=O_{ww}+O_{wb}+O_{bb}$,
\begin{equation*}
ad_{O_E}(\vec{\kappa_F}\cdot\vec{H_F})=0\,,
\end{equation*}
whence the \emph{total number of edges} $N_E=N_{ww}+N_{wb}+N_{bb}$ is conserved by the rewriting rules. Secondly, the commutation relation~\eqref{eq:edgeComm} entails the following evolution equation for the edge-observable EMGF $M(t;\vec{\lambda}_E)$:
\begin{equation*}
\begin{aligned}
\tfrac{\partial}{\partial t} M(t;\vec{\lambda_E})&=\tfrac{\partial}{\partial t} \left.\left\langle \left\vert e^{\vec{\lambda_E}\cdot\vec{O_E}}\right\vert\right. \Psi(t)\right\rangle
\overset{\eqref{eq:MGev}}{=}\sum_{n\geq 1}\tfrac{1}{n!}
\left.\left\langle \left\vert\left(ad_{\vec{\lambda_E}\cdot\vec{O_E}}^{\circ\:n}(\vec{\kappa_F}\cdot\vec{H_F})\right) e^{\vec{\lambda_E}\cdot\vec{O_E}}\right\vert \Psi(t)\right\rangle\right.\\
&=\left.\left\langle \left\vert\left(\kappa_0\left(e^{\lambda_{ww}-\lambda_{wb}}-1\right) H_{0_w} +\kappa_1\left(e^{\lambda_{bb}-\lambda_{wb}}-1\right)H_{1_b}\right)
e^{\vec{\lambda_E}\cdot\vec{O_E}}\right\vert \Psi(t)\right\rangle\right.\,.
\end{aligned}
\end{equation*}
Due to the jump-closure property described in Theorem~\ref{thm:jc} and via~\eqref{eq:VMjc}, we find that 
\begin{equation*}
\bra{}H_{0_w}=\bra{}\bO(h_{0_w})=\bra{}(O_w-1)O_{wb}\,,\qquad \bra{}H_{1_b}=\bra{}\bO(h_{1_b})=\bra{}(O_b-1)O_{wb}\,. 
\end{equation*}
Assuming moreover that $\ket{\Psi(0)}=\ket{G_0}$ for some graph $G_0$ (and whence $\bra{}O_x\ket{\Psi(t)}=N_x$ for $x\in\{w,b\}$ and all $t\geq 0$), we finally obtain the following instance of the Combinatorial Conversion Theorem:
\begin{equation*}
\begin{aligned}
\tfrac{\partial}{\partial t} M(t;\vec{\lambda_E})
&=
K\left(\bar{\kappa}_0\left(e^{\lambda_{ww}-\lambda_{wb}}-1\right)+\bar{\kappa}_1\left(e^{\lambda_{bb}-\lambda_{wb}}-1\right)\right)\tfrac{\partial}{\partial \lambda_{wb}}M(t;\vec{\lambda_E})\\
K&:=\kappa_0(N_w-1)+\kappa_1(N_b-1)\,,\;
\bar{\kappa}_0:=\tfrac{\kappa_0(N_w-1)}{K}\,,\;
\bar{\kappa}_1:=\tfrac{\kappa_1(N_b-1)}{K}\,.
\end{aligned}
\end{equation*}
This particular type of partial differential equation may be solved e.g.\ via the semi-linear normal-ordering technique as introduced in~\cite{Dattoli:1997iz,blasiak2005boson,blasiak2011combinatorial} (and recently applied in~\cite{bdp2017} to semi-linear PDEs for chemical reaction systems). Together with the initial condition 
\begin{equation*}
M(0;\vec{\lambda_E})=e^{\lambda_{ww}N_{ww}+\lambda_{wb}N_{wb}+\lambda_{bb}N_{bb}}\,,
\end{equation*}
where $N_{xy}$ denotes the initial number of edges of type $xy$, 
we thus obtain the following closed-form solution:
\begin{equation*}
M(t;\vec{\lambda_E})=e^{\lambda_{ww}N_{ww}+\lambda_{bb}N_{bb}}\left(e^{-Kt}e^{\lambda_{wb}}+\left(1-e^{-K t}\right)\left(\bar{\kappa}_0e^{\lambda_{ww}}+\bar{\kappa}_1e^{\lambda_{bb}}\right)\right)^{N_{wb}}\,.
\end{equation*}
We present in Figure~\ref{fig:VM2} the time-evolution of the first and second cumulants of the graph observables (computed from the corresponding EMGF $C(t;\vec{\lambda_E})=\log(M(t;\vec{\lambda_E}))$). Noticing that the commutation relations presented in~\eqref{eq:edgeComm} precisely fulfill the requirements of the discrete moment bisimulation theorem (Theorem~\ref{thm:DMB}), we may in fact conclude that via an isomorphism $(ww,wb,bb)\cong(W,D,B)$ of ``indexes-to-colors'' the evolution of the edge observable counts is bisimilar to the discrete rewriting system with linear rules
\begin{equation*}
\GRule{\bullet_W}{\kappa_0(N_w-1)}{\bullet_D}\,,\quad
\GRule{\bullet_B}{\kappa_1(N_b-1)}{\bullet_D}\quad\text{and}\quad \mathsf{O}_{discr}=\{\hat{n}_W,\hat{n}_D,\hat{n}_B\}\,.
\end{equation*}
Since for probability distributions of discrete rewriting systems one has the well-known relationship 
\begin{equation*}
P(t;\vec{\lambda})=M(t;\vec{\ln(\lambda)})
\end{equation*}
between probability and moment generating functions, and due to the conservation of the overall number of edges in this rewriting system, it is thus possible to compute the ternary discrete probability distribution plots as presented in Figure~\ref{fig:VM1}. To the best of our knowledge, this appears to be the first such computation of graph observable count probability distributions for stochastic graph rewriting systems of this kind.\\

We conclude our practical application example by illustrating the typical problem of the failure of the polynomial jump-closure property for a given set of graph observables, indicating that in practice the Combinatorial Conversion Theorem yields a valuable guiding criterion for analyzing SGRSs. The problem is exemplified via the following commutator:
\begin{align*}
[O_{ww},H_{01}]&=H_{01}+\rho\left(\tP{%
\vI{1}{1}{fill=white}{}{fill=white}{}
\vI{1}{2}{fill=white}{}{}{}
\vI{1}{3}{fill=white}{}{fill=white}{}
\seI{1}{1}{=}{}{}{}{}
\seI{1}{2}{=}{}{}{}{}
}\right)\,,\qquad H_{01}=\rho(\delta(h_{01}))\,.
\end{align*}%
Since this term would give a contribution to the evolution of the moments of the edge observable $O_{ww}$ according to the Combinatorial Conversion Theorem, we thus face the typical \emph{problem of growing motifs}: taking repeated commutators and applying the jump-closure theorem, the evolution of the moments of $O_{ww}$ is encoded in a PDE that involves observables for larger and larger motifs, and in particular the set of edge observables thus no longer possesses the polynomial jump-closure property required by the Combinatorial Conversion Theorem. It appears to be a fruitful direction for future research to refine our understanding of the rule-algebraic structure of polynomially jump-closed sets of observables, and to understand better which types of approximation schemes might be available in order to develop a notion of approximate polynomial jump-closure.

\section{Application example case study: a toy cryptocurrency model}
\label{sec:BC}

In order to demonstrate the versatility of the discrete bisimulation ideas, let us consider a toy model for a cryptocurrency, loosely inspired by~\cite{tangleWhitePaper}. The state space of the model is a particular kind of tree-like graph structure with four different types of vertices, with transitions as illustrated in Figure~\ref{fig:tmt}: 
an \emph{active ticket} type (marked $\vSquare$), 
a  \emph{de-activated ticket} type (marked $\vSquareSolid$), 
a \emph{transaction} vertex type (marked $\bullet$) and 
a \emph{ledger} type (marked $\vCirc$). The structure is interpreted as follows:
\begin{itemize}
	\item \emph{Active tickets} $\square$ are necessary in order for a \emph{transaction} vertex $\bullet$ to be able to grow the ledger, but their presence also activates the ticket growth and ticket pair growth transitions.
	\item The tickets $\square$ are \emph{de-activated} over time at random; if all tickets at a transaction vertex $\bullet$ are de-activated, that vertex in effect becomes inactive (i.e.\@ can no longer perform any transitions).
	\item \emph{Rearrangements} of tickets between transaction vertices only occur between vertices $\bullet$ with \emph{exactly one active ticket} and vertices $\bullet$ with \emph{strictly more than one ticket} (of which at least one is active).
\end{itemize}
The model describes a type of self-correcting statistical behavior whereby at any given time in the random evolution of the system there is a high probability that one of the transaction vertices $\bullet$ carries the largest number of active tickets $\vSquare$ and thus dominates the ledger growth.\\

As we will demonstrate in the following, this model possesses a set of polynomially jump-closed observables that moreover exhibit the discrete bisimulation phenomenon. However, since the resulting evolution equations are too complex still to be analytically solvable, we will take advantage of the discrete bisimulation in performing \emph{SSA algorithm based simulations} of the associated discrete rewriting system, thereby gaining access to some key dynamical statistical properties of this stochastic rewriting system system. The set of observables we will consider is based on the following rule algebra elements: 
\gdef\tpScale{0.7}
\begin{equation}\label{eq:tmtObs}
\begin{aligned}
o_{\vSquare}&:=\includeFig{1}{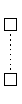}\,,\quad 
o_{\vSquareSolid}:=\includeFig{2}{images/tmtObs.pdf}\,,\quad
o_t:=\includeFig{3}{images/tmtObs.pdf}\,,\quad
o_{\vCirc}:=\includeFig{4}{images/tmtObs.pdf}\\
\\
o_{\bullet_1}&:=\left(\includeFig{5}{images/tmtObs.pdf}\right)\,,\quad
o_{\bullet_g}:=o_{\bullet}-o_{\bullet_1}\,,\quad\;\;\! 
o_{\bullet}:=\includeFig{6}{images/tmtObs.pdf}\\
o_{\vSquareSolid_1}&:=\left(
    \includeFig{7}{images/tmtObs.pdf}
\right)\,,\quad
o_{\vSquareSolid_g}:=o_{\vSquareSolid}-o_{\vSquareSolid_1}\,,\quad 
o_{\vSquareSolid}:=\includeFig{8}{images/tmtObs.pdf}\,.
\end{aligned}
\end{equation}
The canonical representations $\cO_X:=\rho(o_X)$ of the first four rule algebra elements $o_X$ implement the counting of vertices of the different types. The observables $\cO_{\bullet_1}:=\rho(o_{\bullet_1})$ in effect counts transaction vertices $\bullet$ that carry \emph{precisely one} ticket $\vSquare$, since in DPO rewriting the application of a rule that deletes a vertex and some specified incident edges may only be applied to a graph that has this precise pattern as subgraph. Analogously, $\cO_{\bullet_g}:=\rho(o_{\bullet_g})$ counts transaction vertices $\bullet$ with \emph{more than one} ticket (of which at least one is active). This fine distinction will prove rather important for the model to exhibit a discrete bisimulation property, as we will see momentarily. We also note that in all calculations of this type, one typically has to include at least the observables into the analysis that are already present in the system's infinitesimal generator (i.e.\@ in the form of the contribution $\bO(h)$), as is verified by our above choice of observables.

In order to verify that the chosen system of observables possesses the \emph{polynomial jump closure} property (cf.\ Definition~\ref{def:PJC}), we need to calculate the iterated commutators of the various observables with the contributions $h_X$ (ranging over the transitions) to the infinitesimal generator of the CTMC. We collect the corresponding data in Table~\ref{tab:tmtComm}. Note in particular that some of the commutators do not ``close'' polynomially, but the terms that interfere with the closure act trivially on the state space of the model, in which case one may drop the trivially acting terms from the calculation; this is indicated by the notation ${\color{blue}=_S}$ (for equality up to terms acting trivially on the state space). For brevity, we present only these special cases explicitly, since the remaining commutators directly exhibit polynomial jump closure\footnote{As an aside, it might be worth noting that it was crucial to define the rule $h_R$ in this fashion (i.e.\ tickets are only relinked from transaction vertices with precisely one ticket to transaction vertices with strictly more than one ticket), because if one dropped these constraints, one would obtain not only non-closure, but even if one kept the first part of the constraint and would allow relinking to arbitrary transaction vertices, one would lose the discrete bisimulation property.} and are rather simple in structure.\\

Since as the results of Table~\ref{tab:tmtComm} reveal the contributions $h_X$ of the evolution operator are eigenelements under the adjoint action of the vertex observables, it proves convenient to exploit the following lemma in order to proceed:
\begin{lem}\label{lem:tmtAUX}
With the two formal linear operators $\vec{\nu}\cdot\vec{\cO}_V$ and $\vec{\varepsilon}\cdot \vec{\cO}_E$ defined as
\[
	\vec{\nu}\cdot\vec{\cO}_V:=
	\nu_{\vSquare}\cO_{\vSquare}+\nu_{\vSquareSolid}\cO_{\vSquareSolid}+\nu_t \cO_t+\nu_{\vCirc}\cO_{\vCirc}\,,\quad \vec{\varepsilon}\cdot \vec{\cO}_E:=
	\varepsilon_1\cO_{\bullet_1}+\varepsilon_{\vSquareSolid_1}\cO_{\vSquareSolid_1}+\varepsilon_g\cO_{\bullet_g}+\varepsilon_{\vSquareSolid_g}\cO_{\vSquareSolid_g}\,,
\]
and with $H$ the infinitesimal generator of the CTMC, one finds that
\begin{equation}\label{eq:adActLem}
\begin{aligned}
	\sum_{m\geq 1}\frac{1}{m!}\bra{}\left(ad_{\vec{\nu}\cdot \vec{O}_V+\vec{\varepsilon}\cdot \vec{\cO}_E}^{\circ\:m}(H)\right)
	&=\sum_{q\geq 1}\frac{1}{q!} \bra{}\left(ad_{\vec{\nu}\cdot \vec{O}_V}^{\circ\:q}(H)\right)\\
	&\quad+\sum_{r\geq1}\frac{1}{r!} \bra{}\left(ad_{\vec{\varepsilon}\cdot \vec{\cO}_E}^{\circ\:r}(H)\right)\\
	&\quad +\sum_{r,q\geq 1}\frac{1}{r!q!}\bra{}\left(ad_{\vec{\varepsilon}\cdot \vec{\cO}_E}^{\circ\:r}\left(ad_{\vec{\nu}\cdot \vec{O}_V}^{\circ\:q}(H)\right)\right)\,.
\end{aligned}
\end{equation}
\begin{proof}
Taking advantage (in the steps marked $(c)$) of the fact that observables commute (i.e.\ $[\cO_X,\cO_Y]=0$), the claim follows from a straightforward application of the definition of the adjoint action ($e^A B e^{-A}=e^{ad_A}B$ for $A$ a formal operator, cf.\ \cite{hall2015lieGroups}, Prop.~3.35) and using the property $\bra{}H=0$ of the operator $H$ (used in the steps marked $(*)$):
{\allowdisplaybreaks
\begin{align*}
	\sum_{m\geq 1}\frac{1}{m!}\bra{}\left(ad_{\vec{\nu}\cdot \vec{O}_V+\vec{\varepsilon}\cdot \vec{\cO}_E}^{\circ\:m}(H)\right)
	&=\bra{}\left(
		e^{ad_{\vec{\nu}\cdot \vec{O}_V+\vec{\varepsilon}\cdot \vec{\cO}_E}}-1
	\right)H\overset{(*)}{=}\bra{}\left(
		e^{ad_{\vec{\nu}\cdot \vec{O}_V+\vec{\varepsilon}\cdot \vec{\cO}_E}}
	\right)H\\
	&=\bra{}e^{\vec{\nu}\cdot \vec{O}_V+\vec{\varepsilon}\cdot \vec{\cO}_E}H e^{-\vec{\nu}\cdot \vec{O}_V-\vec{\varepsilon}\cdot \vec{\cO}_E}\\
	&\overset{(c)}{=} 
	\bra{}e^{\vec{\varepsilon}\cdot \vec{\cO}_E}
	\left(e^{\vec{\nu}\cdot \vec{O}_V}H e^{-\vec{\nu}\cdot \vec{O}_V}\right)e^{-\vec{\varepsilon}\cdot \vec{\cO}_E}\\
	&= 
	\bra{}e^{\vec{\varepsilon}\cdot \vec{\cO}_E}H e^{-\vec{\varepsilon}\cdot \vec{\cO}_E}
	+\sum_{q\geq 1}\frac{1}{q!}\bra{}e^{\vec{\varepsilon}\cdot \vec{\cO}_E}
	\left(ad_{\vec{\nu}\cdot \vec{O}_V}^{\circ\: q}H \right)e^{-\vec{\varepsilon}\cdot \vec{\cO}_E}\\
	&\overset{(*)}{=}
		\sum_{r\geq 1}\frac{1}{r!}\bra{}\left(
			ad_{\vec{\varepsilon}\cdot \vec{\cO}_E}^{\circ\:r}H
		\right)+\sum_{q\geq 1}\frac{1}{q!}\bra{}
	\left(ad_{\vec{\nu}\cdot \vec{O}_V}^{\circ\: q}H \right)\\
	&\quad \sum_{r,q\geq 1}\frac{1}{r!q!}\bra{}
	\left(
		ad_{\vec{\varepsilon}\cdot \vec{\cO}_E}^{\circ\:r}
		\left(
			ad_{\vec{\nu}\cdot \vec{O}_V}^{\circ\: q}H
		\right)
	\right)\,.
\end{align*}%
}
\end{proof}
\end{lem}
For the concrete model in hand, we find according to Table~\ref{tab:tmtComm} the following adjoint action of the vertex observables on $H$:
\begin{equation}\label{eq:adjVAtmt}
\begin{aligned}
	\sum_{q\geq 1}\frac{1}{q!}\bra{}ad_{\vec{\nu}\cdot \vec{O}_V}^{\circ\: q}(H)&=r_D\left(
		e^{\nu_{\vSquareSolid}-\nu_{\vSquare}}-1
	\right)\bra{}H_D+r_G\left(e^{\nu_{\vCirc}}-1\right)\bra{}H_G\\
	&\quad +r_T\left(e^{\nu_{\vSquare}+\nu_t+2\nu_{\vCirc}}-1\right)\bra{}H_T++r_P\left(e^{2\nu_{\vSquare}+\nu_t+2\nu_{\vCirc}}-1\right)\bra{}H_P\,.
\end{aligned}
\end{equation}
For the edge observables, we find
\begin{equation}\label{eq:adjEAtmt}
\begin{aligned}
	\sum_{r\geq 1}\frac{1}{r!}
		\bra{}ad_{\vec{\varepsilon}\cdot \vec{O}_E}^{\circ\: r}(H)&=
		r_D\left(
		e^{\varepsilon_{\vSquareSolid_1}-\varepsilon_{\vSquare_1}}-1
	\right)\bra{}H_{D_1}
	+r_D\left(
		e^{\varepsilon_{\vSquareSolid_g}-\varepsilon_{\vSquare_g}}-1
	\right)\bra{}H_{D_g}\\
	&\quad +r_T\left(e^{\varepsilon_1}-1\right)\bra{}H_T
	+r_P\left(e^{2\varepsilon_g}-1\right)\bra{}H_P\\
	&\quad +r_R\left(e^{\varepsilon_g-\varepsilon_1}-1\right)\bra{}H_R\,.
\end{aligned}
\end{equation}

The next step in the analysis then consists in assembling the EMGF evolution equation. To this end, we also need to calculate the following identities (where $H_X:=\rho(h_x)$ and $\cO_Y:=\rho(o_Y)$):
\begin{equation}\label{eq:tmtJC1}
\begin{aligned}
	\bra{}H_D&=\bra{}\cO_{\vSquare}\\
	\bra{}H_G&=\bra{}H_T=\bra{}H_P=\bra{}(\cO_{\bullet_1}+\cO_{\bullet_g})\\
	\bra{}H_{D_g}&=\bra{}\cO_{\bullet_g}\,,\; \bra{}H_{D_1}=\bra{}\cO_{\bullet_1}\,.
\end{aligned}
\end{equation}
The calculation of the contribution of the term $\bra{}H_R$ to the EMGF evolution equation is considerably more involved. In fact, we have to invoke the notion of equality up to terms acting trivially on the state space of the model yet again in order to be able to proceed with the analysis. As a first step, we compute (with the operation $\squplus$ as introduced in~\eqref{def:sup})
\begin{equation}\label{eq:tmtJC2}
\begin{aligned}
	\bra{}H_R&=\bra{}\left(\cO_{\bullet_g}  \squplus  \cO_{\bullet_1}\right)
	=\bra{}(\cO_{\bullet} \squplus\cO_{\bullet_1})
	-\bra{}(\cO_{\bullet_1} \squplus  \cO_{\bullet_1})\,.
\end{aligned}
\end{equation}
In order to resolve this equation, we need the following compositions:
\begin{equation}\label{eq:tmtJC3}
\begin{aligned}
	\cO_{\bullet}\cO_{\bullet_1}&=
	\cO_{\bullet} \squplus  \cO_{\bullet_1}+\cO_{\bullet_1}
	+\rho\left(
		\includeFig{1}{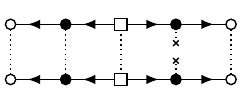}
	\right)
	+\rho\left(
		\includeFig{2}{images/tmtJC3.pdf}
	\right)\\
	&\quad +\rho\left(
		\includeFig{3}{images/tmtJC3.pdf}
	\right)\; {\color{blue}=_S}\; \cO_{\bullet} 
    \squplus  \cO_{\bullet_1}+\cO_{\bullet_1}\\
	\cO_{\bullet_1}\cO_{\bullet_1}&=\cO_{\bullet_1} \squplus  \cO_{\bullet_1}
    +\cO_{\bullet_1}+\rho\left(
		\includeFig{4}{images/tmtJC3.pdf}
	\right)
	+\rho\left(
		\includeFig{5}{images/tmtJC3.pdf}
	\right)\\
	&\quad +\rho\left(
		\includeFig{6}{images/tmtJC3.pdf}
	\right)\; {\color{blue}=_S}\;
	 \cO_{\bullet_1} \squplus  \cO_{\bullet_1}+\cO_{\bullet_1}\\
	 \\
	 \Rightarrow\quad
	 \cO_{\bullet_g} \squplus  \cO_{\bullet_1}&{\color{blue}=_S} \cO_{\bullet_g}\cO_{\bullet_1}\,.
\end{aligned}
\end{equation}
Thus invoking the up-to-trivial action on states equality ${\color{blue}=_S}$ ultimately is what renders the system effectively polynomially jump-closed.\\

The final part of the calculation regards determining the formal differential operator $\bD$ that enters the evolution equation for the EMGF of the chosen observable. Let us introduce the notation
\begin{equation}\label{eq:evoTMT1}
\begin{aligned}
\cM(t;\vec{\nu},\vec{\varepsilon})&:=\bra{}e^{\vec{\nu}\cdot\vec{\cO}_V+\vec{\varepsilon}\cdot \vec{\cO}_E}\ket{\Psi(t)}\\
\vec{\nu}\cdot\vec{\cO}_V&:=\nu_{\vSquare}o_{\vSquare}
+\nu_{\vSquareSolid}o_{\vSquareSolid}
+\nu_t o_t
+\nu_{\vCirc}o_{\vCirc}\,,\quad \vec{\varepsilon}\cdot \vec{\cO}_E:=\varepsilon_1\cO_{\bullet_1}+\varepsilon_g\cO_{\bullet_g}\,.
\end{aligned}
\end{equation}
Then by assembling the results on the commutators as presented in Table~\ref{tab:tmtComm} and the resulting expressions calculated in equations~\eqref{eq:adjVAtmt}, \eqref{eq:adjEAtmt}, \eqref{eq:tmtJC1}, \eqref{eq:tmtJC2} and~\eqref{eq:tmtJC3}, we may determine via Lemma~\ref{lem:tmtAUX} the evolution equation as follows:
\begin{equation*}
\begin{aligned}
	&\sum_{q\geq 1}\frac{1}{q!}\bra{}ad^{\circ\:q}_{\vec{\nu}\cdot\vec{\cO}_V+\vec{\varepsilon}\cdot \vec{\cO}_E}(H)\\
	&{\color{blue}=_S}\;
		r_D\left(
			e^{\nu_{\vSquareSolid}-\nu_{\vSquare}+\varepsilon_{\vSquareSolid_1}-\varepsilon_1}-1\right)\bra{}\cO_{\bullet_1}
			+r_D\left(
			e^{\nu_{\vSquareSolid}-\nu_{\vSquare}+\varepsilon_{\vSquareSolid_g}-\varepsilon_g}-1\right)\bra{}\cO_{\bullet_g}\\
		&\quad + r_G\left(e^{\nu_{\vCirc}}-1\right)\bra{}(\cO_{\bullet_1}+\cO_{\bullet_g})\\
		&\quad + r_T\left(e^{\nu_{\vSquare}+\nu_t+2\nu_{\vCirc}+\varepsilon_1}-1\right)\bra{}(\cO_{\bullet_1}+\cO_{\bullet_g})\\
		&\quad + r_P\left(e^{2\nu_{\vSquare}+\nu_t+2\nu_{\vCirc}+2\varepsilon_g}-1\right)\bra{}(\cO_{\bullet_1}+\cO_{\bullet_g})\\
		&\quad + r_R\left(e^{\varepsilon_g-\varepsilon_1}-1\right)\bra{}\cO_{\bullet_1}\cO_{\bullet_g}\,.
\end{aligned}
\end{equation*}
Here, we have made use of the special property of our model that ticket vertices $\vSquare$ are always attached to transaction vertices, which entails that
\begin{equation}\label{eq:TMTobsEq}
\bra{}\cO_{\vSquare}{\color{blue}=_S\;}\bra{}\cO_{\bullet}=\bra{}\cO_{\bullet_1}+\bra{}\cO_{\bullet_g}\,.
\end{equation}
Moreover, we have made repeated use of the auxiliary elementary identity
\[
	(e^a-1)+(e^b-1)+(e^a-1)(e^b-1)=e^{a+b}-1\,.
\]

Since thus the system possesses the property of polynomial jump closure (Definition~\ref{def:PJC}), we may apply the combinatorial conversion theorem (Theorem~\ref{thm:CCT}), and finally arrive at the evolution equation
\begin{equation}\label{eq:tmtEvo}
\begin{aligned}
	\tfrac{\partial}{\partial t}\cM(t;\vec{\nu},\vec{\varepsilon})&{\color{blue}=_S\;}\bD(\vec{\nu},\vec{\varepsilon},\partial_{\vec{\nu}},\partial_{\vec{\varepsilon}})\cM(t;\vec{\nu},\vec{\varepsilon})\\
	\bD(\vec{\nu},\vec{\varepsilon},\partial_{\vec{\nu}},\partial_{\vec{\varepsilon}})&=
	\bigg[
		r_D\left(e^{\nu_{\vSquareSolid}-\nu_{\vSquare}+\varepsilon_{\vSquareSolid_1}-\varepsilon_1}-1\right)
		+r_G\left(e^{\nu_{\vCirc}}-1\right)\\
		&\qquad
			+ r_T\left(e^{\nu_{\vSquare}+\nu_t+2\nu_{\vCirc}+\varepsilon_1}-1\right)
			+ r_P\left(e^{2\nu_{\vSquare}+\nu_t+2\nu_{\vCirc}+2\varepsilon_g}-1\right)
	\bigg]\partial_{\varepsilon_1}\\
	&\quad
	+\bigg[
		r_D\left(e^{\nu_{\vSquareSolid}-\nu_{\vSquare}+\varepsilon_{\vSquareSolid_g}-\varepsilon_g}-1\right)
		+r_G\left(e^{\nu_{\vCirc}}-1\right)\\
	&\qquad 
		+ r_T\left(e^{\nu_{\vSquare}+\nu_t+2\nu_{\vCirc}+\varepsilon_1}-1\right)
		+ r_P\left(e^{2\nu_{\vSquare}+\nu_t+2\nu_{\vCirc}+2\varepsilon_g}-1\right)
	\bigg]\partial_{\varepsilon_g}\\
		&\quad + r_R\left(e^{\varepsilon_g-\varepsilon_1}-1\right)\partial_{\varepsilon_1}\partial_{\varepsilon_g}\,.
\end{aligned}
\end{equation}
Upon closer inspection, we find that this dynamical system fulfills all requirements of the discrete moment bisimulation theorem (Theorem~\ref{thm:DMB}) if we either
\begin{itemize}
	\item[(a)] only consider \emph{vertex observables} (and use the equivalence ${\color{blue}=_S\;}$ described in~\eqref{eq:TMTobsEq} in order to write $\bra{}(\cO_{\bullet_1} + \cO_{\bullet_g}){\color{blue}=_S\;}\bra{}\cO_{\vSquare}$ in~\eqref{eq:evoTMT1}, which then allows to express the operator $\bD$ in terms of partial derivatives $\partial_{\nu_{\vSquare}}$), or
	\item[(b)] consider all observables but the observable $\cO_{\vSquare}$, i.e.\ set $\nu_{\vSquare}=0$.
\end{itemize}
Since the total number of ticket type vertices $\vSquare$ may also be recovered as the sum of the counts of the observables $\cO_{\bullet_1}$ and $\cO_{\bullet_g}$ (again due to the specific structure of the model's state space), we will from hereon consider option $(b)$. A similar argument of course applies to the inactive ticket type vertices $\vSquareSolid$, so we additionally set $\nu_{\vSquareSolid}=0$ without loss of information. Thus the dynamical evolution of the statistical moments of the graph observables under the transitions of the graph rewriting system coincides with the evolution of the number count observables for the following \emph{discrete} rewriting system (where we have made the isomorphism of observables clear by letting a ``discrete species'' $X_Y$ carry the identifier $Y$ of the associated graph observable $\cO_G$):
\begin{equation}\label{eq:tmtDGRS}
\begin{aligned}
	X_{\bullet_1}&\xrightharpoonup{r_D} X_{\vSquareSolid_1}
	& X_{\bullet_g}&\xrightharpoonup{r_G} X_{\vSquareSolid_g}\\
	\\
	X_{\bullet_1}&\xrightharpoonup{r_G} X_{\bullet_1}+X_{\vCirc}
	& 
	X_{\bullet_g}&\xrightharpoonup{r_G} X_{\bullet_g}+X_{\vCirc}\\
	\\
	X_{\bullet_1}&\xrightharpoonup{r_T}2X_{\bullet_1}+X_t+2X_{\vCirc}
	&
	X_{\bullet_g}&\xrightharpoonup{r_T}X_{\bullet_g}+X_{\bullet_1}+X_t+2X_{\vCirc}\\
	\\
	X_{\bullet_1}&\xrightharpoonup{r_P}X_{\bullet_1}+2X_{\bullet_g}+X_t+2X_{\vCirc}\qquad 
	& 
	X_{\bullet_g}&\xrightharpoonup{r_P}3X_{\bullet_g}+X_{\bullet_1}+X_t+2X_{\vCirc}\\
	\\
	X_{\bullet_1}&+X_{\bullet_g}\xrightharpoonup{r_R}2X_{\bullet_g}\,.
\end{aligned}
\end{equation}
In this highly non-trivial example of a discrete moment bisimulation, one may observe a couple of interesting structures in the associated discrete rewriting system. At first glance, one might in hindsight almost have anticipated some of the discrete transitions listed in~\eqref{eq:tmtDGRS}, since they in a certain sense amount to the \emph{causally consistent modifications} of the observable counts due to the transitions of the original graph rewriting system. But what is fascinating is that our theorem ensures that this set of discrete transitions is in fact \emph{complete}, in the sense that there are no other possibilities for changes of graph observables influencing each other during transitions (at least not until considering a larger set of observables in the original system). It is also worthwhile noting that the transition system described in~\eqref{eq:tmtDGRS} is \emph{not} of the type of a branching process, due to the second order transitions described in the last line of~\eqref{eq:tmtDGRS}. One might thus envision that there exists an intimate link between discrete moment bisimulations and causal analysis of stochastic rewriting systems, which we plan to investigate further in future work.\\

Back to the system at hand, one may finally take full advantage of the discrete moment bisimulation and extract dynamical statistical information on graph observables of our model via \emph{simulation} of the associated discrete rewriting system. The standard approach for the simulation of discrete rewriting system is using Gillespie's stochastic simulation algorithm (SSA)~\cite{Gillespie_1977}. We chose to apply the algorithm via using the \textsc{KaSim} simulation suite\footnote{\href{https://kappalanguage.org}{Kappa Language (https://kappalanguage.org):} A rule-based language for modeling interaction networks}. Referring to Appendix~\ref{sec:kappa} for the concrete implementation of our model (with an example for a choice of reaction rates given), one may produce system trajectories such as the one displayed in Figure~\ref{fig:TMTtraj}. Further work with the simulation engine would allow one to extract statistical information such as means of observable counts etc.\@ in the exact same fashion as familiar from the numerical study of chemical reaction networks. As the exemplary computation presented in this section illustrates, one would of course ultimately strive to develop some form of automatized algorithms in order to facilitate the search for interesting closure properties of rule-algebraic commutation relations, an endeavor which we leave to future work.

\begin{figure}
\centering
\includegraphics[width=0.9\textwidth]{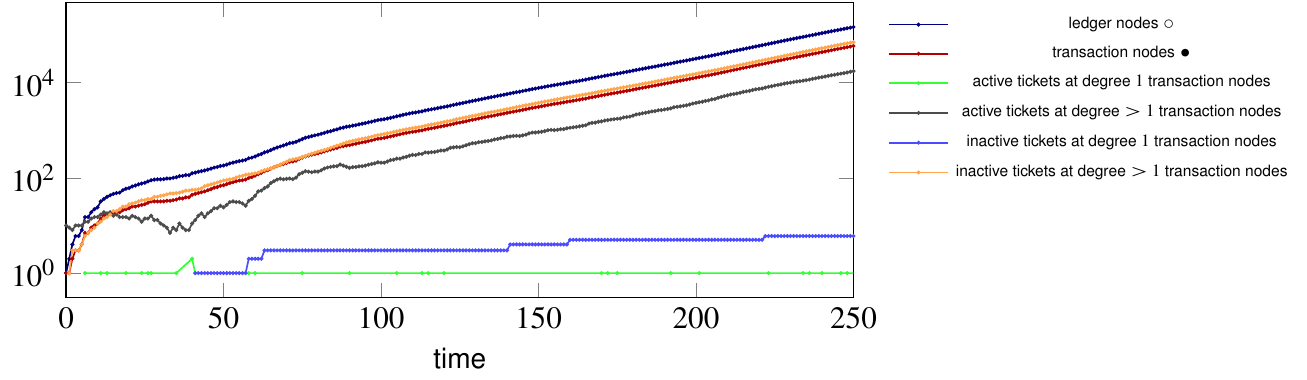}
\caption{Example trajectory as produced by simulating the Kappa model as described in~\ref{sec:kappa}. \label{fig:TMTtraj}}
\end{figure}

\section{Conclusion and outlook}
\label{sec:CO}

The framework presented in this paper firmly establishes the theory of stochastic rewriting systems (SRSs) as a natural sub-discipline of the study of continuous-time Markov chains and of statistical physics. Our formulation of SRSs in terms of a rule-algebraic stochastic mechanics framework possesses a number of marked practical advantages over previous approaches. At a fundamental level, it renders computations of the evolution of statistical quantities such as the moments of observables quite analogous in nature to comparable computations e.g.\ in the theory of \emph{chemical reaction systems (CRSs)} (see e.g. \cite{bdp2017}). Of immediate practical importance is the discovery that just as in the case of CRSs, the dynamics of statistical moments of observables is determined by certain \emph{static algebraic relations}, with the precise formulation of this discovery given by the Combinatorial Conversion Theorem in Section~\ref{sec:CC}. The theorem hinges on our novel and generalized notion of \emph{polynomial jump-closure} of sets of observables, whose rather combinatorial nature inspired the moniker of the theorem. Perhaps the most fascinating aspect of this mathematical structure is a striking resemblance of the moment evolution equations to certain evolution equations of other types of generating functions in the theory of analytical combinatorics and of combinatorial species~\cite{FlajoletSedgewick,blasiak2011combinatorial}. As a first hint of the utility of exploring such structural similarities further, we have demonstrated in Section~\ref{sec:AE} how to apply a technique known as semi-linear normal-ordering~\cite{Dattoli:1997iz,blasiak2005boson,blasiak2011combinatorial} (originally developed as part of the study of exponential generating function of series of combinatorial numbers) in order to obtain exact, closed-form solutions to moment generating function evolution equations. 

The passage to the study of exponential generating functions of moments of graph observables results in an additional novel avenue for the analysis of SRSs, which we refer to as \emph{moment bisimulation} (Section~\ref{sec:DB}). Contrary to more traditional approaches of bisimulations for CTMCs which are based on notions of lumpability of the state spaces of the CTMCs~\cite{buchholz_1994,Feret2012137,Petrov_2012}, our notion of bisimulation focuses instead on the evolution equations for the statistical moments of observables: given two sets of polynomially jump-closed observables for two different SRSs, we define the two systems to be moment bisimilar (w.r.t.\@ the given sets of observables) if their moment evolution equations according to the Combinatorial Conversion Theorem coincide. While a detailed study of the structure of the resulting bisimilarity classes is left to future work, we have presented one particularly interesting such class in the form of the \emph{Discrete Bisimulation Theorem}, which states the precise constructive conditions under which an SRS with a chosen set of polynomially jump-closed observables is bisimilar to a class of SRSs on discrete graphs, better known as chemical reaction systems. The criterion is purely rule-algebraic and combinatorial in nature, and we presented in Section~\ref{sec:AE} a concrete non-trivial example of an SRS that exhibits this phenomenon. We believe that moment bisimulation has immense potential as a technique to improve the tractability of the analysis of SRSs both in theory and in practice, and plan on taking it as the basis for the development of combinatorics-based approximation algorithms in future work.

We envision that the stochastic mechanics approach of the present paper in conjunction with further detailed studies of the techniques of combinatorial species theory~\cite{FlajoletSedgewick,blasiak2011combinatorial} will ultimately lead to a deeper understanding of the evolution of stochastic dynamical rewriting systems. In particular, our prototypical application example presented in Section~\ref{sec:CC} hints at the possibility for taking the rule-algebraic commutation relations that govern the moment-dynamics as a constructive criterion in the design of SRSs. Taking inspiration from frameworks such as the Kappa language~\cite{danos2004formal} developed for describing biochemical reaction networks, one might hope to unveil a fine interplay between the data structures (the objects to be rewritten), the precise rule sets (modeling the transitions) and the observables (the patterns to be observed) for a given SRS, where alternative design choices for a given system may very well strongly influence the tractability of the stochastic moment evolutions. A necessary precondition for such developments would consist in a constructive characterization of the  moment bisimilarity classes beyond the discrete bisimulation class, which we leave for future work.

\paragraph{Acknowledgments.}

The work of N.B.\@ was supported by a \emph{Marie Sk\l{}odowska-Curie Individual fellowship} (Grant Agreement No.~753750 -- RaSiR). The work of I.G.\@ was supported by the ANR grant REPAS. We would like to thank G.~Dattoli, G.H.E.~Duchamp, J.~Krivine, P.A.M.~Melli\`{e}s, K.A.~Penson and N.~Zeilberger for fruitful discussions.

\begin{figure}
\subfigure[Time-evolution of the edge observable count probability distribution]{\label{fig:VM1A}
\begin{minipage}{\textwidth}
  \begin{minipage}{.45\textwidth}
  \includegraphics[width=\textwidth,page=2]{images/VM/tern-plot-VM-25-5-0-50-0-1-20-1o2-1o18.pdf}
\end{minipage}
\hspace{.025\textwidth}
\begin{minipage}{.45\textwidth}
  \includegraphics[width=\textwidth,page=6]{images/VM/tern-plot-VM-25-5-0-50-0-1-20-1o2-1o18.pdf}
\end{minipage}\\
\begin{minipage}{.45\textwidth}
  \includegraphics[width=\textwidth,page=11]{images/VM/tern-plot-VM-25-5-0-50-0-1-20-1o2-1o18.pdf}
\end{minipage}
\hspace{.025\textwidth}
\begin{minipage}{.45\textwidth}
  \includegraphics[width=\linewidth,page=21]{images/VM/tern-plot-VM-25-5-0-50-0-1-20-1o2-1o18.pdf}
\end{minipage}
\end{minipage}
}\\
\subfigure[Time-evolution of the means and (co-)variances of the edge observables]{\label{fig:VM2}
  \begin{minipage}{.45\linewidth}
    \includegraphics[width=\linewidth]{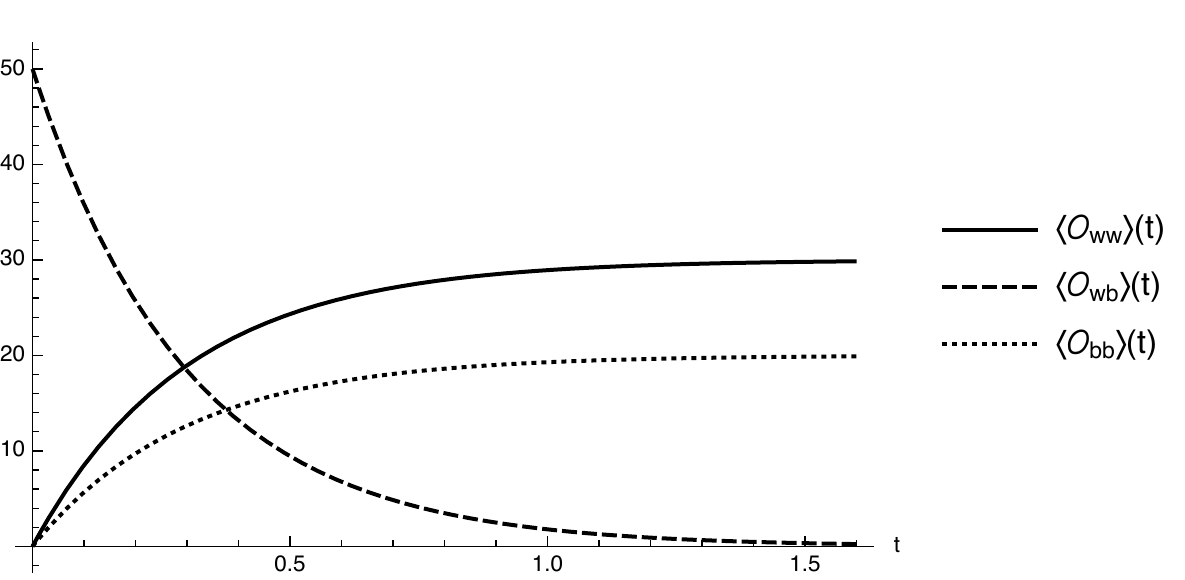}
  \end{minipage}
    \hspace{.05\linewidth}
  \begin{minipage}{.45\linewidth}
    \includegraphics[width=\linewidth]{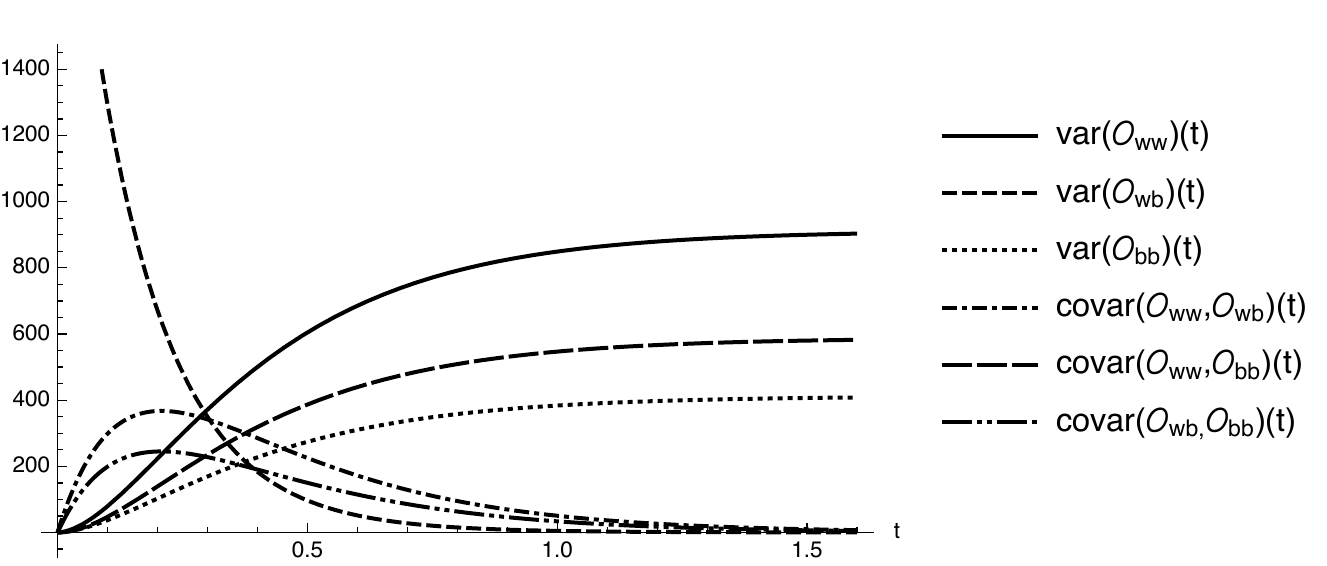}
  \end{minipage}
}%
\caption{Exemplary time-evolution of a Voter Model with ``edge-flipping'' rules only, for $\kappa_0=\tfrac{1}{2}$, $\kappa_1=\tfrac{1}{18}$ and an initial state $\ket{\Psi(0)}=\ket{G_0}$ with $(N_w,N_b,N_{ww},N_{wb},N_{bb})=(5,25,0,50,0)$.}
  \label{fig:VM1}
\end{figure}

\clearpage



\newpage
\appendix

\section{Kappa language implementation of the discrete bisimilar reaction model}\label{sec:kappa}

\begin{lstlisting}
////signatures 
// (here: somewhat trivial, since we are only dealing with discrete vertices)

%agent: TA()			// transaction node type
%agent: LN()			// ledger node type
%agent: PT(v{o g},T{a i})	// pattern with transaction node 
                        	// of degree one (o) or greater (g) 
                        	// and with ticket 
                        	// active (a) or inactive (i)

////variables 

//base rates:
%var: 'rD' 0.42		// base rate of ticket deactivation per second
%var: 'rG' 0.05		// base rate of ledger growth per second
%var: 'rT' 0.16		// base rate of ticket growth per second
%var: 'rP' 0.05		// base rate of ticket pair growth per second
%var: 'rR' 0.2		// base rate of ticket rearrangement per second

//initial condition parameters:
%var: 'nT0' 10.0		// initial number of active tickets
%var: 'nt0' 1.0			// initial number of transaction nodes
%var: 'nL0' 1.0			// initial number of ledger nodes

////observables
%obs: 'ledger nodes' |LN()|
%obs: 'transaction nodes' |TA()|
%obs: 'active tickets at degree 1 transaction nodes' |PT(v{o},T{a})|
%obs: 'active tickets at degree >1 transaction nodes' |PT(v{g},T{a})|
%obs: 'inactive tickets at degree 1 transaction nodes' |PT(v{o},T{i})|
%obs: 'inactive tickets at degree >1 transaction nodes' |PT(v{g},T{i})|

////initial conditions
/* NOTE: this entails that the initial state is one transaction vertex with nT0 active tickets and one leddger vertex linked to the transaction vertex */
%init: 'nT0' PT(v{g},T{a})	
%init: 'nt0' TA()
%init: 'nL0' LN()

////transitions
'rule D1' PT(v{o},T{a}) -> PT(v{o},T{i}) @ 'rD'
'rule D2' PT(v{g},T{a}) -> PT(v{g},T{i}) @ 'rD'

'rule G1'  PT(v{o},T{a}), .  ->  PT(v{o},T{a}) , LN() @ 'rG'
'rule G2'  PT(v{g},T{a}), .  ->  PT(v{g},T{a}) , LN() @ 'rG'

'rule T1' PT(v{o},T{a}), ., ., ., . -> PT(v{o},T{a}),PT(v{o},T{a}),TA(),LN(),LN() @ 'rT'
'rule T2' PT(v{g},T{a}), ., ., ., . -> PT(v{g},T{a}),PT(v{o},T{a}),TA(),LN(),LN() @ 'rT'

'rule P1' PT(v{o},T{a}), ., ., ., ., . -> PT(v{o},T{a}),PT(v{g},T{a}),PT(v{g},T{a}),TA(),LN(),LN() @ 'rT'
'rule P2' PT(v{g},T{a}), ., ., ., ., . -> PT(v{g},T{a}),PT(v{g},T{a}),PT(v{g},T{a}),TA(),LN(),LN() @ 'rT'

'rule rR' PT(v{o},T{a}),PT(v{g},T{a}) -> PT(v{g},T{a}),PT(v{g},T{a}) @ 'rR'
\end{lstlisting}

\newpage

\renewcommand\arraystretch{1.0}
\begin{longtable}{RCL}
  \caption{Nested adjoint action of observables on contributions to the infinitesimal generator.\label{tab:tmtComm}}\\
\toprule
ad_{o_X}h_Y &=& o_X \ast_{\cR} h_Y-h_Y\ast_{\cR}o_X\\
    \midrule
\endfirsthead
\multicolumn{3}{l}{{\tablename\ \thetable\  \textit{-- continued from previous page}}} \\
\toprule\\
ad_{o_X}^{\circ \:q}h_Y &=& o_X \ast_{\cR} h_Y-h_Y\ast_{\cR}o_X\\
    \midrule
\endhead
\bottomrule\\ 
\multicolumn{3}{l}{\textit{Continued on next page}} \\
\endfoot
\bottomrule\\
\endlastfoot
ad_{\vec{\nu}\cdot \vec{o}_V}(h) &=& 
ad_{\nu_{\vSquare}o_{\vSquare}
+\nu_{\vSquareSolid}o_{\vSquareSolid}
+\nu_t o_t
+\nu_{\vCirc}o_{\vCirc}}(r_D h_D+r_G h_G\\
&&\qquad \qquad +r_T h_T+r_P h_P+r_Rh_R)\\
&=& (\nu_{\vSquareSolid}-\nu_{\vSquare})r_Dh_D
+\nu_{\vCirc} r_G h_G\\
&&\quad +(\nu_{\vSquare}+\nu_t+2\nu_{\vCirc})r_Th_T
+(2\nu_{\vSquare}+\nu_t+2\nu_{\vCirc})r_Ph_P\\
\midrule
\left[o_{\bullet_1},h_D\right]&=& 
\left[
	\includeFig{1}{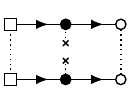},
    \includeFig{2}{images/tmtComm.pdf}
\right]
=-\includeFig{3}{images/tmtComm.pdf}
=:-h_{D_1}\\
\\
\left[o_{\bullet},h_D\right]&=& 
\left[
	\includeFig{4}{images/tmtComm.pdf},
    \includeFig{5}{images/tmtComm.pdf}
\right]
    =-\includeFig{6}{images/tmtComm.pdf}
    =:-h_{D_{\bullet}}\\
	\\
	\Rightarrow\quad
	\left[\varepsilon_1 o_{\bullet_1}+\varepsilon_g o_{\bullet_g},h_D\right]
	&=& -\varepsilon_1 h_{D_1}-\varepsilon_g h_{D_g}\,,\qquad  h_{D_g}:=h_{D_{\bullet}}-h_{D_1}\\
	\\
	\left[\varepsilon_1 o_{\bullet_1}+\varepsilon_{\bullet} o_{\bullet},h_{D_1}\right]
	&=& \varepsilon_1\left(
	-h_{D_1}
    -\includeFig{7}{images/tmtComm.pdf}
	\right)\\
	&&\qquad \qquad +\varepsilon_{\bullet}\left(
	-h_{D_1}
    -\includeFig{8}{images/tmtComm.pdf}
	\right)\\
	&{\color{blue}=_S}& -\varepsilon_1 h_{D_1}-\varepsilon_{\bullet} h_{D_1}\\
    \\
	\left[\varepsilon_1 o_{\bullet_1}+\varepsilon_{\bullet} o_{\bullet},h_{D_{\bullet}}\right]
	&=& \varepsilon_1\left(
	-h_{D_1}
    -\includeFig{9}{images/tmtComm.pdf}
	\right)\\
	&&\qquad \qquad\\
	&&\qquad \qquad+\varepsilon_{\bullet}\left(
	-h_{D_{\bullet}}
    -\includeFig{10}{images/tmtComm.pdf}
	\right)\\
	&{\color{blue}=_S}& -\varepsilon_1 h_{D_1}-\varepsilon_{\bullet} h_{D_{\bullet}}\\
	\\
	\Rightarrow\quad
	ad_{\varepsilon_1 o_{\bullet_1}+\varepsilon_{g} o_{\bullet_g}}^{\circ\: q}\left(h_{D}\right)
	&{\color{blue}=_S}& (-\varepsilon_1)^q h_{D_1}+(-\varepsilon_g)^q h_{D_g}\\
	\midrule
	ad_{\varepsilon_{\vSquareSolid_1} o_{\vSquareSolid_1}
	+\varepsilon_{\vSquareSolid_g} o_{\vSquareSolid_g}}^{\circ\: q}\left(h_{D}\right)
	&{\color{blue}=_S}& (\varepsilon_{\vSquareSolid_1})^q h_{D_1}+(\varepsilon_{\vSquareSolid_g})^q h_{D_g}\\
	\midrule
	ad_{\varepsilon_1 o_{\bullet_1}+\varepsilon_{g} o_{\bullet_g}+\varepsilon_{\vSquareSolid_1} o_{\vSquareSolid_1}
	+\varepsilon_{\vSquareSolid_g} o_{\vSquareSolid_g}}\left(h_G\right) &=& 0\\
	\midrule
	ad_{\varepsilon_1 o_{\bullet_1}+\varepsilon_{g} o_{\bullet_g}+\varepsilon_{\vSquareSolid_1} o_{\vSquareSolid_1}
	+\varepsilon_{\vSquareSolid_g} o_{\vSquareSolid_g}}\left( h_T\right) &=& 
	\varepsilon_1 r_T h_T\\
	\midrule
	ad_{\varepsilon_1 o_{\bullet_1}+\varepsilon_{g} o_{\bullet_g}+\varepsilon_{\vSquareSolid_1} o_{\vSquareSolid_1}
	+\varepsilon_{\vSquareSolid_g} o_{\vSquareSolid_g}}\left(h_P\right) &=& 
	2\varepsilon_g r_Ph_P\\
	\midrule
		ad_{\varepsilon_1 o_{\bullet_1}+\varepsilon_{g} o_{\bullet_g}+\varepsilon_{\vSquareSolid_1} o_{\vSquareSolid_1}
	+\varepsilon_{\vSquareSolid_g} o_{\vSquareSolid_g}}\left(h_R\right) &=& 
		(-\varepsilon_1+\varepsilon_g)h_R\\
\end{longtable}

\begin{figure}[htbp]
\centering\gdef\tpScale{0.9}\gdef\imgScale{1.2}
\subfigure[Ticket de-activation]{\label{fig:Drule}
	$\vcenter{\hbox{\includegraphics[scale=\imgScale]{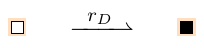}}}\qquad \widehat{=}\qquad h_D:=\tP{\vI{1}{1}{sqVertSolidGlow}{}{sqVertGlow}{}}\,,\quad o_D:=\tP{\vI{1}{1}{sqVert}{}{sqVert}{}}$
}\\
\subfigure[Ledger growth]{\label{fig:Grule}
	$\vcenter{\hbox{\includegraphics[scale=\imgScale]{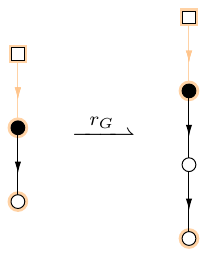}}}\qquad \widehat{=}\qquad 
	h_G:=\vcenter{\hbox{\includegraphics[page=1]{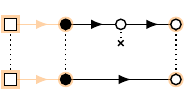}}}\,,\quad 
		o_G:=\includeFig{2}{images/Grule.pdf}$
}\\
\subfigure[Ticket production]{\label{fig:Trule}
	$\vcenter{\hbox{\includegraphics[scale=\imgScale]{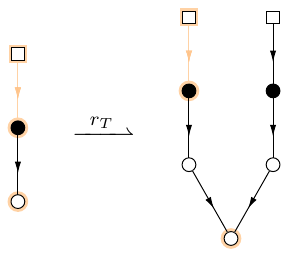}}}\qquad \widehat{=}\qquad \begin{array}{l}
	h_T:=\includeFig{1}{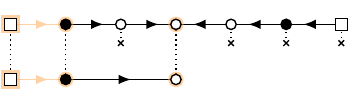}\\
	o_T:=\includeFig{2}{images/Trule.pdf}
		\end{array}$
}\\
\subfigure[Ticket pair production]{\label{fig:Prule}
	$\vcenter{\hbox{\includegraphics[scale=\imgScale]{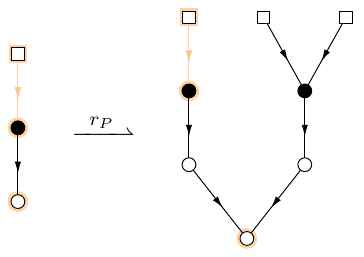}}}\qquad \widehat{=}\qquad \begin{array}{l}
	h_P:=\includeFig{1}{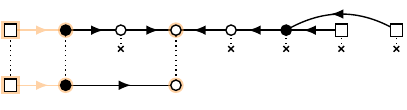}\\ 
	o_P:=\includeFig{2}{images/Trule.pdf}
		\end{array}$
}
\\
\subfigure[Ticket rearrangement]{\label{fig:Rrule}
	$\vcenter{\hbox{\includegraphics[scale=\imgScale]{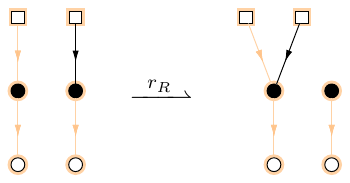}}}\qquad \widehat{=}\qquad \begin{array}{l}
	h_R:=\includeFig{1}{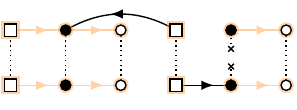}\\ 
		\qquad \qquad \qquad -
		\includeFig{2}{images/Rrule.pdf}\\
	o_R:=\left(\includeFig{3}{images/Rrule.pdf}
		-
		\includeFig{4}{images/Rrule.pdf}\right) \squplus  \left(
		\includeFig{5}{images/Rrule.pdf}
		\right)
		\end{array}$
}%
\caption{Transitions, corresponding rule algebra elements and associated observable rule algebra elements for the crypto-currency toy model. For the transitions, orange highlights indicate the graphical elements that are (effectively) preserved throughout the transition.\label{fig:tmt}}
\end{figure}

\end{document}

%% file: BDG-LMCS--preamble.tex
\usepackage{amsmath,amssymb}
 \usepackage{mathabx} 
\usepackage{bbold}
\usepackage{mathtools}
\usepackage{etoolbox}
\usepackage{xifthen}
\usepackage{scalerel}

\usepackage{subfigure}
\usepackage{array}
\usepackage{wrapfig}
\usepackage{graphicx}

\usepackage{pdflscape}
\usepackage{afterpage}
\usepackage{thmtools}
\usepackage{thm-restate}
\usepackage{longtable}
\usepackage{booktabs}

\gdef\tpScale{0.7}

\newcommand{\tP}[1]{\ensuremath{\vcenter{\hbox{\begin{tikzpicture}[transform shape, scale=\tpScale]\normalsize#1\end{tikzpicture}}}}}

\gdef\colSEP{0.8}

\colorlet{h1color}{blue!40!black} 
\colorlet{h2color}{orange!90!black} 
\colorlet{h3color}{blue!40!white} 
\colorlet{h4color}{green!40!black} 

\makeatletter
\def\bstr{b}
\def\bfstr{bf}
\def\cstr{c}
\def\fstr{f}
\def\lst{A,B,C,D,d,E,F,G,H,I,J,K,L,M,N,O,P,Q,R,S,T,U,V,W,X,Y,Z}

\newcommand{\MkB}[1]{\expandafter\def\csname\bstr#1\endcsname{\mathbb{#1}}}
  \@for\i:=\lst\do{%
    \expandafter\MkB \i     }
    
\newcommand{\MkBF}[1]{\expandafter\def\csname\bfstr#1\endcsname{\mathbf{#1}}}
  \@for\i:=\lst\do{%
    \expandafter\MkBF \i     }

\newcommand{\MkCal}[1]{\expandafter\def\csname\cstr#1\endcsname{\mathcal{#1}}}
  \@for\i:=\lst\do{%
    \expandafter\MkCal \i     }
    
\newcommand{\MkFrak}[1]{\expandafter\def\csname\fstr#1\endcsname{\mathfrak{#1}}}
  \@for\i:=\lst\do{%
    \expandafter\MkFrak \i     }
    
\makeatother

\renewcommand{\vec}[1]{\underline{#1}}

\newcommand{\mono}[1]{\mathsf{mono}(#1)}

\newcommand{\obj}[1]{\mathsf{obj}(#1)}

\newcommand{\Lin}[1]{\mathsf{Lin}(#1)}

\newcommand{\MatchSet}[2]{\mathsf{M}_{#1}(#2)}

\newcommand{\comp}[3]{#1 \stackrel{#2}{\blacktriangleleft} #3}

\newcommand{\GRule}[3]{#1\xleftharpoonup{#2}#3}

\newcommand{\ntC}{\circledast}

\newcommand{\bra}[1]{\left\langle #1\right\vert}
\newcommand{\ket}[1]{\left\vert #1\right\rangle}
\newcommand{\braket}[2]{\left\langle \left. #1 \right\vert #2\right\rangle}

\newcommand{\Prob}{\operatorname{\mathsf{Prob}}}

\newcommand{\End}{\operatorname{\mathsf{End}}}

\newcommand{\Stoch}{\operatorname{\mathsf{Stoch}}}

\newcommand{\rSpan}[5]{\!\!\begin{mycdInline}[ampersand replacement=\&] (#1 \& #3 \ar[l,hook',"#2" description] \ar[r,hook,"#4" description] \& #5) \end{mycdInline}\!\!}

\newcommand{\rSpanAlt}[5]{\left(#1\xhookleftarrow{#2}#3\xhookrightarrow{#4}#5\right)}

\newcolumntype{L}{>{$}l<{$}} 

\newcolumntype{R}{>{$}r<{$}} 

\newcolumntype{C}{>{$}c<{$}} 

\newcommand{\OneVertG}[1][]{%
\ifthenelse{\equal{#1}{}}{%
\tP{%
\node[vertices] (a) at (1,1) {};}}{%
\tP{%
\node[vertices,#1] (a) at (1,1) {};}}}

\newcommand{\TwoVertG}[1][]{%
\ifthenelse{\equal{#1}{}}{%
\tP{%
\node[vertices] (a) at (1,1) {};
\node[vertices] (b) at (1.5,1) {};}}{%
\tP{%
\node[vertices,#1] (a) at (1,1) {};
\node[vertices,#1] (b) at (1.5,1) {};}}}

\newcommand{\TwoVertGL}[3][]{%
\ifthenelse{\equal{#1}{}}{%
\tP{%
\node[vertices] (a) at (1,1) {#2};
\node[vertices] (b) at (1.5,1) {#3};}}{%
\tP{%
\node[vertices,#1] (a) at (1,1) {#2};
\node[vertices,#1] (b) at (1.5,1) {#3};}}}

\newcommand{\TwoVertEdgeG}[1][]{%
\ifthenelse{\equal{#1}{}}{\tP{%
\node[vertices] (a) at (1,1) {};
\node[vertices] (b) at (1.5,1) {};
\draw (a) edge (b);}}{\tP{%
\node[vertices,#1] (a) at (1,1) {};
\node[vertices,#1] (b) at (1.5,1) {};
\draw (a) edge[#1] (b);}}}

\newcommand{\TwoVertDirEdgeG}[1][]{%
\ifthenelse{\equal{#1}{}}{\tP{%
\node[vertices] (a) at (1,1) {};
\node[vertices] (b) at (1.5,1) {};
\draw (a) edge[dirEdge] (b);}}{\tP{%
\node[vertices,#1] (a) at (1,1) {};
\node[vertices,#1] (b) at (1.5,1) {};
\draw (a) edge[#1,dirEdge] (b);}}}

\usepackage{environ}

\def\temp{&} \catcode`&=\active \let&=\temp

\gdef\mycdScale{1}
\NewEnviron{mycd}[1][]{%
\begin{tikzpicture}[baseline=(current bounding box.west),
transform shape, scale=\mycdScale]
\node[inner sep=0, outer sep=0] {\begin{tikzcd}[#1]
\BODY
\end{tikzcd}};
\end{tikzpicture}}

\gdef\mycdinlineScale{1.0}
\NewEnviron{mycdInline}[1][]{%
\begin{tikzpicture}[baseline=-0.5ex, 
transform shape, scale=\mycdinlineScale]
\node[inner sep=0, outer sep=0] {\begin{tikzcd}[#1]
\BODY
\end{tikzcd}};
\end{tikzpicture}}

\usepackage{listings}

\definecolor{Brown}{rgb}{0.65, 0.16, 0.16}
\definecolor{ForestGreen}{rgb}{0.13, 0.55, 0.13}
\definecolor{Frame}{RGB}{55,155,198}
\definecolor{Red}{rgb}{0.93, 0.11, 0.14}
\definecolor{RoyalBlue}{RGB}{65,105,225}

\input{Kappa_LaTeX_Lang}

\newcommand{\includeFig}[2]{\vcenter{\hbox{\includegraphics[page=#1]{#2}}}}

\gdef\tpGOScale{0.5}

 \newcommand{\tPgo}[1]{\ensuremath{\vcenter{\hbox{\begin{tikzpicture}[transform shape, scale=\tpGOScale]\normalsize#1\end{tikzpicture}}}}}

 \def\vCirc{\scalerel*{\includegraphics[page=1]{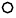}}{b}}
 \def\vSquare{\scalerel*{\includegraphics[page=2]{images/bcSymbols.pdf}}{b}}
 \def\vSquareSolid{\scalerel*{\includegraphics[page=3]{images/bcSymbols.pdf}}{b}}

 \usepackage{tikzexternal}

 \tikzsetexternalprefix{tikz/}


%% file: Kappa_LaTeX_Lang.tex
\lstdefinelanguage{kappa}{
	alsoletter={\%,:,-,<,>,@},
	morekeywords=[1]{\%agent\:, \%token\:, \%var\:, \%obs\:},
	morekeywords=[2]{\%init\:},
	morekeywords=[3]{\%mod\:, \%def\:, \$ADD, \$DEL, \$SNAPSHOT, \$STOP, \$FLUX, \$TRACK, \$UPDATE, \$PLOTENTRY, \$PRINT, \$PRINTF},
	morekeywords=[4]{->, <->, @},
	morestring=[d]{'},
	morecomment=[l]\#
}
